%% file: main.tex
\documentclass[letterpaper,11pt]{article}

\usepackage[utf8]{inputenc} 
\usepackage[T1]{fontenc}
\usepackage{url}              
\usepackage{ifthen}          
\usepackage{cite}             
\usepackage{graphicx}
\usepackage[cmex10]{amsmath}  
\usepackage{amssymb}
\usepackage{amsthm}
\usepackage{tikz, calc}
\usepackage{caption}
\usepackage{subcaption}

\newcommand{\X}{\mathbf{X}}
\newcommand{\Y}{\mathbf{Y}}
\newcommand{\G}{\mathbf{G}}
\newcommand{\B}{\mathbf{B}}

\include{headers}

\title{Group Testing with Correlation under Edge-Faulty Graphs}

\begin{document}
\newcommand*\samethanks[1][\value{footnote}]{\footnotemark[#1]}

\author{Hesam Nikpey,\thanks{
  Email: {{\small{\texttt{\{hesam,jungyeol,xingranc,swati,saeedi\}@seas.upenn.edu.}}}}}
  \and Jungyeol Kim, \and Xingran Chen, \and Saswati Sarkar, \and Shirin Saeedi Bidokhti}
\date{}

\maketitle

\begin{abstract}
In applications of group testing in networks, e.g. identifying individuals who are infected by a disease spread over a network, exploiting correlation among network nodes provides fundamental opportunities in reducing the number of tests needed.
We model and analyze group testing on $n$ correlated nodes whose interactions are specified by a graph $G$. We model correlation through an edge-faulty random graph formed from $G$ in which each edge is dropped with probability $1-r$, and all nodes in the same component have the same state.  

We consider three classes of graphs: cycles and trees, $d$-regular graphs and stochastic block models or SBM, and obtain lower and upper bounds on the number of tests needed to identify the defective nodes. Our results are expressed in terms of the number of tests needed when the nodes are independent and they are in terms of $n$, $r$, and the target error. In particular, we quantify the fundamental improvements that exploiting correlation offers by the ratio between the total number of nodes $n$ and the equivalent number of independent nodes in a classic group testing algorithm.

The lower bounds are derived by illustrating a strong dependence of the number of tests needed on the expected number of components.
In this regard, we establish a new approximation for the distribution of component sizes in ``$d$-regular trees'' which may be of independent interest and leads to a lower bound on the expected number of components in $d$-regular graphs.

The upper bounds are found by forming dense subgraphs in which nodes are more likely to be in the same state. When $G$ is a cycle or tree, we show an improvement by a factor of $log(1/r)$. For grid, a graph with almost $2n$ edges, the improvement is by a factor of ${(1-r) \log(1/r)}$, indicating drastic improvement compared to trees. When $G$ has a larger number of edges, as in SBM, the improvement can scale in $n$.


\end{abstract}

\section{Introduction}

   Group testing \cite{du2000combinatorial} is a well studied problem at
   the intersection of many fields, including computer science
   \cite{dorfman1943detection,cheraghchi2011group, coja2020optimal, cheraghchi2021semiquantitative, cuturi2020noisy},
   information theory \cite{aldridge2019group,jadbabai, goodrich2008improved}  and computational biology
   \cite{knill1995group, farach1997group}. The goal is to find an
   unknown subset of $n$ items  that are different from the rest using
   the least number of tests. The target subset is often referred to
   as \emph{defective}, corrupted or infected, depending on the field of study.
   In this work, we use the term defective.
   To find the subset of defectives, items are tested in groups.
   The result of a test is
   positive if and only if at least one item in the group is defective.

   Over the years, this problem has been formulated via two approaches:
   the combinatorial approach and the information theoretic approach.
   In the ``combinatorial'' version of the problem, it is assumed that
   there are $d$ defective items that are to be detected  with zero error
   \cite{du2000combinatorial}. Using adaptive group testing (i.e., when who to test next depends on the results of the previous tests), there is a
   matching upper and lower bound on the number of tests in the form $d
   \log n + O(d)$ \cite{du2000combinatorial}. Using non-adaptive group
   testing (i.e., when the testing sequence is pre-determined), there is an upper bound of $O(d^2 \log (n/d))$ and an almost
   matching lower bound of $\Omega(\frac{d^2 \log n}{\log d})$. The
   ``information theoretic'' approach, on the other hand, assumes a prior
   statistic on the defectiveness of items, i.e., item $i$ is assumed to
   be defective with probability $p_i$. The aim in this case is to
   identify  the defective set
   with  high probability \cite{li2014group}. Roughly speaking,
   there is a lower bound in terms of the underlying entropy of the
   unknowns, and an almost  matching upper bound up to a $\log n$ factor
   of the lower bound.

  In most existing works, it is assumed that the state of the items,
  whether or not they are defective,  are independent of each other, which is not realistic in many applications. Group testing, for example, can identify the infected individuals using fewer tests, and therefore in a more timely manner, than individual testing, during the spread of an infectious disease (eg, COVID-19)
   \cite{brault2021group,verdun2021group, mutesa2021pooled, aldridge2020conservative, gollier2020group}. But the infection state of individuals are in general correlated, with correlation levels ranging from high to low, depending on how close they live: same household (high), same neighborhood, same city, same country (low). Correlation levels also depend on other factors such as frequency of contact, the number of paths between the individuals in
  the network of interactions. We elaborate on this further in Section \ref{sec:model}.  
{Another example is the multiaccess channel problem: here a number of users want to communicate through a common channel and we want to assign time slots to them to avoid any conflicts. Before assigning, we aim to find the number of active users that want to send a message. Using group testing, we can identify the number of active users faster by asking a subset of users if any of them is active \cite{berger1984random,chlebus2001randomized,goodrich2008improved}. But again, nodes are often not independent. Generally, some subset of users might communicate among themselves more often and hence, be more correlated.} With this motivation, we aim to model such  correlation, design group testing techniques that  exploit it, and  quantify the gain they provide in reducing the number of tests needed.

The closet to our work are \cite{ahn2021adaptive, arasli2021group, nikolopoulos2021group, nikolopoulos2021groupp} where specific correlation models are considered and group testing methods are designed and analyzed. In \cite{ahn2021adaptive}, {the authors consider correlation that is imposed by a one day spread of an infectious disease in a clustered network modeled by a stochastic block model. Each node is  initially defective (infected) with some probability and in the next time step,  its neighbors  become defective probabilistically. The authors provide a simple adaptive algorithm and prove optimality in some regimes of operation and under some assumptions.} 
In \cite{arasli2021group}, the authors model correlation through a {random edge-faulty} graph $G$. Each edge $e$ is {realized} in the graph with a {given} probability $r_e$. So depending on how the graph is realized, it is partitioned into some connected components. Each connected component is assumed defective with probability $p$ {(in which case, all the nodes in that  component are defective)} and otherwise  non-defective  with probability $1-p$. The authors focus {only on a subset of the realizations by studying the case} in which the set of connected
components across realizations forms a particular nested structure.  
More specifically, they only consider a subset of realizations such that for each realization, it is instantiated from another realization in the subset, i.e., two realizations have the same components except one component that is partitioned into two components. Only one realization does not obey this rule, the one with the least number connected components. 
 They found a near optimal non-adaptive testing strategy for any distribution over the considered realizations of the graph and showed optimality for specific graphs.  

{The correlation model we consider is close to} the work of \cite{arasli2021group}. We consider a {random (edge-faulty}) graph $G$ { where each edge is realized} with probability $r$. {In a realized graph,}  each connected component is assumed defective with probability $p$. {As opposed to \cite{arasli2021group}},  we {do not constrain our study to a subset of the realizations and instead consider all possible realizations of the graph $G$}.  {Despite its simplicity, our model  captures  two key features}. First, given a network of nodes, one expects that when a path between two nodes gets shorter, {the probability that they are in the same state increases. Our proposed model captures this intuition.}  By adding one edge to a path, the probability of being in the same state reduces by a factor of $r$. Second, two nodes might have long distances from each other, but when there are many edge-distinct paths between them, it is more likely that they are in the same state. Under our model, by adding distinct paths between two nodes,   the probability of them being in the same state increases. 



   In other related works, a graph could represent potential constraints
   on testing (among independent nodes) \cite{cheraghchi2012graph,spang2018unconstraining}. This can be viewed as a complementary direction to our setting in which  a graph models the inherent dependency of the nodes, but there is no constraint in how the groups can be formed for testing. 
   In \cite{cheraghchi2012graph}, the authors  design optimal
   non-adaptive testing strategies when each group is constrained to be path
   connected. In particular, they use random walks to obtain the pool of
   tests.
   In a follow up work, \cite{spang2018unconstraining}  shows that either
   the constraints are too strong and no algorithm can do better than
   testing most of the nodes, or optimal algorithms can be designed
   matching with the unconstrained version of the problem.
   This is attained by sampling each edge with probability $r$ ($0<r<1$ and
   optimized). If a  component is large enough, the algorithm tests the
   entire component.
   Our approach in this paper has similarities with
   \cite{spang2018unconstraining} in aiming to find parts of the graph
   that are  large and connected enough so that they remain connected
   with a decent probability after realizing the edges, but our techniques to find the dense subgraphs and the corresponding analysis are different.

\subsection{Our Model}\label{sec:model}
   We start by motivating the key  attributes that we capture in our model through an example of testing for infections  in a network of people (nodes). Consider the interaction  network for the spread of an infectious disease (e.g. COVID-19) in a network of people/nodes. 
   There is an edge between two nodes if the corresponding individuals are in physical proximity  for a
   minimum amount of time each week. Such individuals are more likely to be in the same state than those who have been distant throughout. Thus, firstly, the probability of being in the same state decreases with increase in the length of paths (i.e., distance in interaction network) between nodes.
   Second, infection is more likely to spread from one node to another if there are
   many distinct paths between them. Thus, the probability that two nodes are in the same state increases with the increase in the number of distinct paths between them.

   We capture correlation through a faulty-edge graph model. Consider a
   graph $G = (V,E)$ where the node set $V$, $|V|=n$, represents the items and the
   edge set $E$ represents  connections/correlations between them.
   Suppose each edge  is  realized with probability $0 \leq r \leq 1$.
   After the sampling, we have a random graph that we denote by $G_r$.
   Each
   node is either defective or non-defective. All
   nodes in the same component of $G_r$ are in the same state, rendering defectiveness a component property. We consider that each
   component is defective with probability (w.p.)  $p$ independent  of others.  As an
   example, consider graph $G$ with five nodes and eight edges, and a sampled graph realization  $G_r$ as shown in Figure~\ref{example:partition} (left) and
   Figure~\ref{example:partition} (right) respectively.
   When $r=1/3$, $G_r$ is realized w.p.
   $(\frac{1}{3})^3 (\frac{2}{3})^5$. There are two components in $G_r$, namely, $v_1, v_4,
   v_5$ and $v_2, v_3$; $v_1, v_4, v_5$ are in the same state, which is defective w.p. $p$, independent
   of the state of $v_2, v_3.$
   
   \textcolor{black}{Two nodes are guaranteed to be in the same state in $G_r$ if there exists at least one path that connects them in $G_r$. The probability that a path in $G$ survives in $G_r$ increases with increase in $r$. Thus both the parameter $r$ and the graph $G$ determine the correlation between states of different nodes; the correlation is higher if $r$ is higher, states of all nodes are independent for $r=0$, while the correlation is the highest possible for a given $G$ for $r=1.$}

   This model importantly captures the two attributes we discussed: Clearly, a long path
   between two nodes in $G$ has a smaller chance of survival in $G_r$, compared to a
   short path, making the end nodes less likely to be in the same state
   as the length of the path in $G$ between them increases. Moreover,  the probability that at least one path between two nodes survive in $G_r$ increases with increase in the number of distinct paths between them in $G$, so having distinct paths between a pair of nodes in $G$ makes them more
   likely to be in the same state.
  
   We aim to find the minimum expected number of tests needed to find the defective items with at most $\epsilon n$ errors, where $\epsilon$  can potentially be of order $o(1)$.
To be precise, let $\textsc{\#ERR(H)}$ be the number of nodes mispredicted by an algorithm on graph $H$. Then we require to have 
\begin{align}\mathbb{E}_{H \sim G_r}[\#ERR(H)] \leq \epsilon n
\label{eq:error}
\end{align} 
where the expectation is taken over $G_r$ and possible randomization of the algorithm. We refer to it as the \emph{average} error.
\begin{remark}\label{remark:differror} 
The definition of error in classic probabilistic group testings such as  \cite{li2014group} is a stronger notion of error probability where the goal is to correctly predict  all nodes with probability $1-\epsilon$, and with probability $\epsilon$ one or more nodes are mispredicted. This is stronger than our definition of average error in \eqref{eq:error} because with  probability $\epsilon$ at most $n$ nodes are mispredicted in the classic group testing, so the average error would be less than $\epsilon n$, the allowed error in our model.
\end{remark}
\noindent We mostly work with the notion of  average error in this paper. In the last section (Section~\ref{sec:generalized}), we consider a stronger notion of error to limit the 
\emph{maximum} error:  the group testing schemes now need  to upper bound the number of mispredicted nodes by $\epsilon n$ with high probability. We recover all the results of the paper for this stronger constraint on error as well.

Methodologically, we relate the problem of group testing with correlation to that in a network with fewer nodes in which the states of nodes are independent. We obtain bounds on the number of group tests required in the former, to satisfy the constraints we consider on errors, in terms of known bounds in the latter. The relative quantification provides a basis for comparison and determination of the improvements that can be obtained  by exploiting  correlation.  

\begin{figure*}
\centering
\input{graph_partition1}%
\input{graph_partition2}
    \caption{Graph $G$ after a realization of edges}
    \label{example:partition}
\end{figure*}
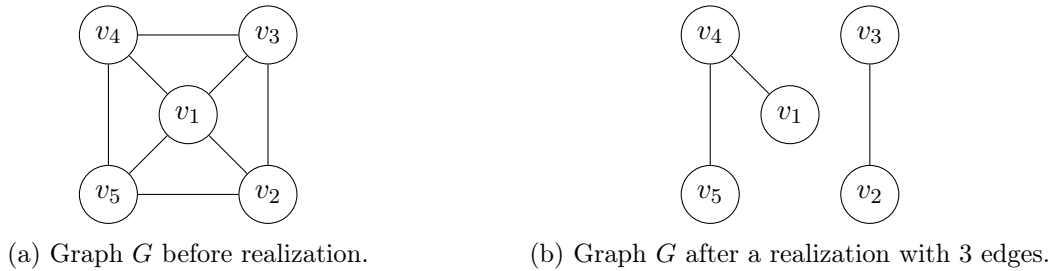

   The tests can not be designed with the knowledge of $G_r$, only the value of $r$ is known apriori. In the extreme case of $r = 0$, the problem is reduced to the classic
   group testing with $|V|=n$ independent nodes. In the extreme case of $r
   = 1$, all components of $G$ remain connected and have the same state and hence the problem is reduced \textcolor{black}{to testing a single node.}
   When $0 <r < 1$, the problem is non-trivial, because there can be multiple components, some with more than one node, and the number and composition of the components is apriori unknown. Thus, it is not apriori known which nodes will be in the same state. Our group testing strategies will seek to circumvent this challenge by identifying parts of $G$ that are
   connected enough so that they remain connected in $G_r$ with a high
   probability.

\subsection{Contributions}
We obtain upper and lower bounds on the number of group tests needed to determine the states, defective or otherwise,  of individual nodes in a large class of interaction graphs in the presence of correlation among the states of nodes. We progressively consider 1)  cycles and trees  (about $n$ links), 2) $d-$regular graphs (about $dn/2$ links) and 3) stochastic block models or SBM (with potentially $\Theta(n^2)$ links). The correlation is captured by the factor $r$ (see Section \ref{sec:model}). The bounds are obtained in terms of the number of tests needed when  the states are independent, and help us quantify the efficiency brought forth by group testing in terms of $r.$ \textcolor{black}{In particular, our group testing strategies exploit correlation and build upon classical group testing strategies,  but with fewer number of nodes. The ratio between this number and the total number of nodes ($n$) determines the benefits of correlation in our strategies and we refer to it as the (multiplicative) improvement factor. Note that this is not the ratio between the number of tests that are needed, but the ratio between the number of nodes that a classical group testing algorithm gets as input. As such, our results are valid for any group testing algorithm, and they can be translated to the ratio between the total number of tests,  accordingly, as needed.}

For trees and cycles, we prove an upper bound on the optimal number of tests in terms of the number of  group tests when there are $n \log (1/r)$ independent nodes. Note that one can trivially determine the states  of each node by disregarding correlation and testing among $n$ nodes (e.g. using classic group testing techniques). Our upper bound therefore shows that group testing can reduce the tests. The (multiplicative) improvement factor of $\log(1/r)$ is meaningful (less than 1) when $r > 1/2$.  As $r$ approaches $1$ the multiplicative factor reduces even further implying even greater benefits due to group testing. Our lower bound, on the other hand, shows an improvement factor $(1-r)$. 

For $d-$regular graphs,  we prove new results for the distribution of components. This leads to a lower bound that is expressed as a sum series depending on $r$ and $n$. We further prove an upper bound for a specific $4-$regular graph, namely grid, in terms of the number of group tests when there are $n {(1-r) \log(1/r)}$ independent items. Thus, the improvement  factor is  ${(1-r) \log(1/r)}$, as opposed to only $\log(1/r)$ for trees; this hints us that group testing gets drastically more efficient  for denser graphs. 

The stochastic block model divides the network into communities such that nodes in the same community are more connected than nodes in different communities. We show that the reduction in the test count due to group testing can be classified into three regimes:  1) strong intra-community connectivity but sparse inter-community connectivity, which reduces the  effective number of independent nodes to the number of communities, 2) strong connectivity leading to an (almost) fully connected graph, in which case all nodes have the same state and one independent test is representative 3) sparse connectivity leading to many isolated nodes, in which case the states of all nodes are independent. The first case reduces to independent group testing with the number of nodes equal to the number of communities, second regime needs a constant number of tests, and finally the third regime reduces to independent group testing with $n$ nodes. \textcolor{black}{The tight upper and lower bounds that are known in the literature for the independent group testing subsequently apply for the first and third cases; for the second case, the analysis is rather simple as there is only a constant number of tests needed.}

\subsection{Our Methods and Ideas}

\textcolor{black}{We now briefly describe the mathematical techniques that we follow to obtain the bounds. The techniques constitute a contribution in themselves as a graphical structure was not investigated earlier for group testing except under significant restrictions as described earlier. For the upper bound for a cycle, we divide the cycle $G$ into subgraphs of size $l$ ($l$ nodes) where $l$ is a parameter. The subgraphs are connected in $G$, but need not be connected in $G_r$ which we do not know apriori. For every subgraph, we select a node which we consider as representative of the subgraph and determine the states of the representatives of all the subgraphs using group  testing strategies deployed when states of nodes are independent (ie, we do not exploit possible correlation between the representatives). We consider the state of each node in each subgraph as that of the representative; this is indeed the case if each subgraph is connected, otherwise the state of some nodes are determined in error. The probability of each subgraph being connected decreases with increase in $l$, thus the expected number of errors, which can be computed as a function of $r, l, n$,  increases with increase in $l.$ The number of representatives and therefore the number of nodes subjected to group tests described above is $n/l$. Thus the number of group tests is non-increasing with $l.$ Thus $l$ represents a tradeoff between the expected number of errors and the number of group tests, and $l$ is selected appropriately to ensure a low number of group tests subject to ensuring that the expected number of errors does not exceed the specified limit. The number of group tests for the $l$ that satisfies the specified error constraints provides the upper bound on the number of tests for a cycle.}

\textcolor{black}{The upper bound for an arbitrary tree can be obtained similarly, with the additional significant complication that for an arbitrary tree $G$ and an $l$ that satisfies constraints on error one may not be able to obtain subgraphs in $G$ of size $l$ that are connected in $G$ (in contrast when $G$ is  a cycle,  a path of size $l$ constitutes such a subgraph). Refer to Figure~\ref{fig:star5} for an illustration of this challenge. We get around this challenge by using subgraphs that are not connected in $G$ by themselves, but become connected in $G$ through at most $l$ additional nodes in $G$ (which are not in the subgraph). Construction of  such subgraphs is not apriori clear and constitutes an innovation needed for upper bounding the number of tests needed for trees, above and beyond the overall methodology. Each such subgraph is connected in $G_r$ if the links in $G$ among these $2l$ nodes survive in $G_r$, the probability of this event can again be expressed as a function of $l, r$, and as before this probability decreases with increase in $l.$ The rest of the methodology is similar to for cycles.}

\textcolor{black}{The overall methodology for obtaining the upper bound for grids is the same as that for cycles. The subgraphs in question constitute sub-grids of $l$ nodes. We determine the probability that each such sub-grid is connected in $G_r$ through a recursive decomposition which is not apriori obvious.}

\textcolor{black}{The characterization of lower bound on the number of tests needed for cycle or trees constitutes another innovation. To obtain the lower bound one can assume the knowledge of the components of $G_r$, which one does not know in reality. Nodes in each component of $G_r$ have the same state, which is independent of those of the states of nodes in other components. Thus each component can be considered a super-node and the states of the super-nodes are independent of each other. Thus, one needs at least as many tests as that when there are $C(G_r)$ nodes whose states are independent, where $C(G_r)$ is the number of components of $G_r$. A lower bound can now be obtained if the random variable $C(G_r)$ can be bounded. We accomplish this objective by observing that number of components in a stochastic graph constitutes an edge-exposure martingale and the value of this random variable is concentrated around its expectation, courtesy of Azuma’s inequality which holds for such martingales.}

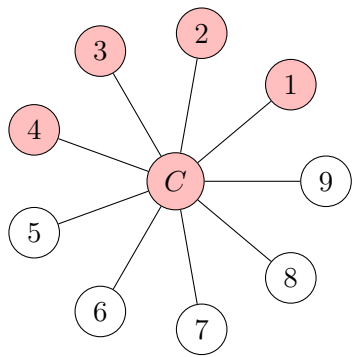
\begin{figure}
\centering
\input{star}%
    \caption{A star can not be partitioned into trees of size $l$ for any $ l > 2.$ This is because any partition of a star of $n$ nodes into smaller trees constitutes of a star of $m$ nodes, for a $m$ of our choice and $n-m$ isolated nodes. The figure shows an example where $n = 10, m=5.$ Nodes $C, 1, 2, 3, 4$ constitute the star, nodes $5, \ldots 9$ constitute isolated nodes. We can however partition the star into subgraphs of size $l$ (ie, having $l$ nodes) each of which can be connected through the central node even when the node does not belong in the subgraph. The figure shows such a partition with $l=5.$ The subgraph consisting of  nodes $C, 1, 2, 3, 4$ is connected, the rest of the $5$ nodes constitutes a subgraph that can be connected through the central node.}
    \label{fig:star5}
\end{figure}

\textcolor{black}{We initially present the above results for constraints on the expected number of errors. We subsequently generalize them for a stronger constraint, that on the number of errors with a high probability  (rather than expectation thereof), for cycle and trees. Lower bounds can be obtained following the same methodology as before, except that lower bounds on the number of tests needed for independent nodes for this stronger constraint do not exist in the literature. We derive the latter lower bound first, adapting some existing proof techniques relying on Information-theoretic inequalities. We show that the lower bounds can be improved for some specific structures of $G$, such as when $G$ is a star. This stronger notion of error allows us to include the structure of $G$ more heavily and directly in this proof, rather than going via an analysis for the  number of components. We have resorted to the latter in obtaining the lower bounds for cycle, tree for the weaker constraint on errors. The upper bound can be obtained following the same broad structure, though some additional technical challenges arise. The number of errors over all subgraphs can be bounded with high probability using Hoeffding’s inequality if the number of errors in different subgraphs are independent. This happens for cycle the subgraphs are non-overlapping paths and   the events that they remain connected in $G_r$ are independent. Nonetheless, this does not happen for trees because nodes in one subgraph may be connected through those in other subgraphs. We surmount this technical challenge by invoking an innovative exposure of nodes of the tree such that the number of connected subgraphs constitutes a node exposure Martingale which satisfies the requisite Lipschitz condition for Azuma’s inequality to hold. The high probability bound on the number of errors over all subgraphs now follow via Azuma’s inequality.}

\section{Preliminaries and Notations}\label{sec:prel}
   We use the following notations for the rest of the paper.
   Let $\textsc{CrltOPT}(G, r, p, \epsilon)$ be the expected number of
   tests in an optimal algorithm on graph $G$ with correlation parameter $r$,
   probability of defectiveness $p$, and an error of $\epsilon n$. Let
   $\textsc{IndepOPT}(n, p, \epsilon)$ be the minimum expected number of
   tests needed for $n$  items in order to find the defective set with
   the error probability at most $\epsilon$, where each item is
   \textit{independently} defective with probability $p$. Notice that $\textsc{IndepOPT}(n, p, \epsilon)$ does not depend on $G$. Note that   
   group testing is potentially beneficial when the number of defective items is $o(n)$.
   It is often assumed that the (expected) number of defective items is $n^\alpha, \alpha<1$, so from now on we assume $p = o(1)$.
   It is also noteworthy  that the definitions of error in $\textsc{IndepOPT}$ and $\textsc{CrltOPT}$ are different, as mentioned in Remark~\ref{remark:differror}. 
   When clear from
   the context, we may drop $p, r, \epsilon$ from the notations. We write $A \simeq f(n)$ if $A = f(n) + o(f(n))$. We write $A \lesssim f(n)$ when $A \leq f(n) \pm o(f(n))$

The following lemma provides a lower bound on $\textsc{CrltOPT}$ in terms of $\textsc{IndepOPT}$ -- the minimum number of group tests needed in discovery of the defective set among independent nodes.

\begin{lemma} \label{lem:components}
Let $C(G_r)$ be the number of connected components of $G_r$. Then
\begin{align}
\textsc{IndepOPT}(C(G_r), p, \epsilon n)  \leq \textsc{CrltOPT}(G, r, p, \epsilon). \label{eq:lowerbound}
\end{align}
\end{lemma}
\begin{proof}
\textcolor{black}{Consider the idealized scenario in which one knows the components of $G_r$. Suppose that we have $C(G_r)$ components.  All nodes in the same  component are in the same state, and the states of nodes in different components are independent. Thus, each component can be replaced by one  node and  the minimum number of tests   is that needed to test a graph with  $C(G_r)$ independent nodes  and the expected number of errors is at most $\epsilon n$. This number corresponds to \textsc{IndepOPT}($C(G_r), p, \gamma)$ for some $\gamma$ such that the expected number of errors is at most $\epsilon n$. A classical group testing algorithm with \textsc{IndepOPT}($C(G_r), p, \gamma)$ tests may make (at least) an error in finding one node’s state with probability $\gamma$; thus the expected number of errors is at least $\gamma$. This implies that we need $\gamma \leq \epsilon n.$ Clearly \textsc{IndepOPT}($C(G_r), p, \gamma)$  is a non-increasing function of $\gamma.$ Thus at least \textsc{IndepOPT}($C(G_r), p, \epsilon n)$  tests are needed when one knows the components of $G_r$. If one does not know the components, the expected number of tests can only increase. The lemma follows.}
\end{proof}

\begin{remark}
At first glance, the bound may come across as counter-intuitive because the number of tests in a graph in which the states of nodes are independent provides a lower bound on that when the states are correlated. The apparent contradiction is resolved when we note that the lower bound is obtained in terms of  the number of tests in a graph with a fewer nodes than the original: number of components in $G_r$  instead of the number of nodes in $G_r.$ 
\end{remark}

Lemma \ref{lem:components}  illustrates a connection between the number of connected components and the minimum number of tests \textsc{CrltOPT}.  We next discuss concentration lemmas for a class of graph theoretic functions, including the number of connected components. By having the concentration results, we would be able to replace $C(G)$ by its expectation in \eqref{eq:lowerbound}.
\subsection{Concentration Results}
\begin{definition}\cite{alon2016probabilistic}
A graph theoretic function $f$ is said to satisfy the edge Lipschitz condition if, whenever $H$ and $H^{\prime}$ differ in only one edge, $\left|f(H)-f\left(H^{\prime}\right)\right| \leq 1$ .

\end{definition}

Note that the number of components  $C(G)$ is edge Lipschitz, as when  two graphs differ in only one edge, they either  have the same number of components, or the graph with one less edge has one additional component. One can define a node Lipschitz condition by replacing edge with node  \cite{alon2016probabilistic}.

\noindent \textbf{The Edge Exposure Martingale}. Let $e_1, e_2, \ldots, e_m$ be an arbitrary order of the edges. We define a martingale $X_0, X_1, \ldots, X_m$  where $X_i$ is the value of a graph theoretic function $f(H)$ after exposing $e_1, e_2, \ldots, e_i$. Note that $X_0$ is a constant which is the expected of $f(G)$, where $G$ is drawn from $G_r$. 
This is a special case of martingales sometimes referred to as a Doob martingle process, where $X_i$ is the conditional expectation of $f(H)$, as long as the information known at time $i$ includes the information at time $i-1$ \cite{alon2016probabilistic}. 
The same process can be defined for node exposure martingales, where the nodes are exposed one by one \cite{alon2016probabilistic}. 
Node exposure can be seen as exposing one node at each step, so at the $i^{\text{th}}$ step the graph has $i$ nodes along with the edges between them. You can find more about the topic in \cite[Chapter 7]{alon2016probabilistic}.
We have the following theorem.

\begin{theorem}\cite{alon2016probabilistic}
When $f$ satisfies the edge \textcolor{green}{(resp. node)} Lipschitz condition, the corresponding edge  \textcolor{green}{(resp. node)} exposure martingale satisfies $\left|X_{i+1}-X_{i}\right| \leq 1$.
\end{theorem}


We then have Azuma's inequality.

\begin{theorem}\cite{alon2016probabilistic}\label{theorem:alonconct}
Let $X_{0} = c, \ldots, X_{m}$ be a martingale with
$$
\left|X_{i+1}-X_{i}\right| \leq 1
$$
for all $0 \leq i<m$. Then
$$
\operatorname{Pr}\left[\left|X_{m}-c\right|>\lambda \sqrt{m}\right]<2 e^{-\lambda^{2} / 2}.
$$
\end{theorem}

\begin{corollary} \label{corol: conenctcomp}
Let $\delta > 0$. Then with probability $1 - \delta$ we have 
$$|C(G_r) - \mathbb{E}[C(G_r)|] \leq O( \sqrt{m \log 1/\delta }).$$

\noindent Specifically, in the case  that $\mathbb{E}[{C(G_r)}] = cn$ for a constant $c$, and $G$ has $m = o(n^2)$ edges, then with high probability the number of connected components is within $cn \pm o(n) \sqrt{\log 1/\delta}$.
\end{corollary}
\begin{proof}
Consider the number of components $C(G)$ which is an edge-lipschitz function of the graph. Now define $X_0 = \mathbb{E}[C(G_r)|]$ and $X_m = C(G_r)$. Applying Theorem \ref{theorem:alonconct} concludes the proof. 
\end{proof}
\subsection{Classic Group Testing Results} 

Lemma \ref{lem:components} provides a lower bound on $\textsc{CrltOPT}$ in terms of $\textsc{IndepOPT}$.  Indeed, any lower bound on the latter leads to a lower bound on the former. We now  review some useful lower and upper bounds on $\textsc{IndepOPT}(G, r, p, \epsilon)$ next. 

In the probabilistic group testing of \cite{li2014group}, there are $n$ individuals and every individual $i$ is \emph{independently} defective with probability $p_i$. Let $\mathbf{X}$ be the indicator vector of the {items' defectiveness} and { $H(\mathbf{X}) = \sum_{i} p_i \log \frac{1}{p_i}$ be the entropy of $\mathbf{X}$. The testing can be adaptive or non-adaptive, meaning the testing at each step can depend on the results of the previous tests, or not, respectively. \cite{li2014group} proves the following lower bound on the  required number of tests (for both adaptive and non-adaptive scenarios)}
\begin{theorem}\cite{li2014group} \label{theorem:lowerboundindep}
Any probabilistic group testing algorithm whose probability of error is at most $\epsilon$ requires at least $ (1 - \epsilon) H(\mathbf{X})$ tests.
\end{theorem}


For the upper bounds in this paper, roughly speaking, we partition the graph into groups of nodes and assume that the state of nodes in different groups are independent. Subsequently we obtain bounds for group testing on the original graph in terms of those when state of nodes are independent. We therefore utilize the existing bounds for smaller networks in our context.
{Below, we summarize the  probabilistic group testing results of \cite{li2014group} for upper bounds on number of tests needed when items are independent.} 

\begin{theorem}\cite{li2014group} \label{theorem:adaptive_indep}
There is an adaptive algorithm with probability of error at most $\exp \left(-2 \delta^{2} n^{1 / 4}\right)$ such that the number of tests is less than
$$
2(1+\delta)(H(\mathbf{X})+3 \mathbb{E}[X])
$$
\end{theorem}

\begin{theorem} \cite{li2014group}\label{theorem:non-adaptive_indep}
For any $0< \epsilon' \leq 1$ and $\delta>0$, if the entropy of $\mathbf{X}$ satisfies
$$
H(\mathbf{X}) \geq \Gamma_{\gamma}^{2}
$$
where
$$
\Gamma_{\gamma}:=\log _{2}\left(\log _{1 / \gamma}\left(\frac{2 n}{\epsilon'}\right)\right)
$$
then with probability of error at most
$$
\epsilon \leq \Gamma_{\gamma}^{-\delta+1}+\frac{1}{2} \epsilon'
$$
there is a non-adaptive algorithm that requires no more than
$$
T \leq \frac{e \ln n}{\log _{2}(1 / \gamma)}(1+\delta) H(\mathbf{X})+\Gamma_{\gamma}^{2}+2 \mathbb{E}[X]
$$
tests.
\end{theorem}

\section{Graphs with a Low Number of Edges}
{In this section, we consider graphs that have $n+c$ edges, where $c$ is a constant. Specifically, we prove lower and upper bounds on the number of tests needed when the underlying graph is a cycle or tree, where the lower bound can be generalized to graph with $n+c$ edges. We {obtain these} bounds by formulating the problem into one with independent nodes, potentially with less nodes than the original problem. First, we  use the lemmas introduced in Section~\ref{sec:prel} to obtain lower bounds for cycles and trees. Next, we propose group testing algorithms and prove upper bounds for cycles and trees.}

\subsection{A Lower Bound for Cycles and Trees} \label{sec:lowerbound}
 By  Corollary~\ref{corol: conenctcomp}, if we know the expected number of components in a graph, we  {would be} able to lower bound the minimum number of tests needed by Lemma~\ref{lem:components}. The following theorem proves a lower bound for cycle and trees.

\begin{theorem} \label{theorem:cycletreelowbound}
Let $G$ be a cycle or a tree. Then we have
\begin{align*}
& \textsc{IndepOPT}((1-r)n - 10 \sqrt{n \log n},p, \epsilon n) \leq
 \textsc{CrltOPT}(G, r, p, \epsilon) + O(1/n).
\end{align*}

\end{theorem}

\begin{proof}
 In a tree, by removing each edge we get one more  component, so after removing $k$ edges the tree has $k+1$ components and the cycle has $k$ components. 

Each edge is removed with probability $1-r$, so the expected number of components is $1 + (1-r)(n-1)$ for trees, and $(1-r)n$ for cycles. By Corollary~\ref{corol: conenctcomp}, the number of components is $(1-r)n \pm O(\sqrt{n\log n })$  with probability $1 - 1/n^2$, and with probability $1/n^2$, the difference in tests is at most $n$, hence $O(1/n)$ additional tests. Applying Lemma~\ref{lem:components} thus completes the proof.
\end{proof}

The above proof also works for any graph with $n + c$ edges, where $c$ is a constant. In other words, when the number of edges is less than $n + c$, a lower bound on the number of tests needed for almost $(1-r)n$ independent nodes is also a lower bound on the number of tests needed under our model with correlation $r$.

\subsection{An Upper Bound for Cycles and Trees} \label{sec:cycleandtree}

In this section, we provide algorithms to find the defective set and provide theoretical bounds. We start by considering that $G$ is a simple cycle, and subsequently generalize the ideas to arbitrary trees. Note that after having an algorithm for trees, we would have an algorithm for general graphs, by just considering a tree spanning it. But the algorithm might be far from optimal. 

The general idea is  to partition the graph $G$ into subgraphs that will remain connected in $G_r$ with high probability. The nodes in those connected subgraphs will thus have the same states.  
We can then select a candidate node for each subgraph to be tested. By knowing the probability of each subgraph being connected and the probability of error in classic group testing, we can estimate the error in our problem \textcolor{black}{as a function of the size of the subgraphs and design the subgraphs accordingly.}

First, we provide our results for when $G$ is a cycle.

\begin{theorem}\label{theorem:cycle}
Consider a cycle of length $n$. Let $l = \max \{\frac{\log [1/(1- \epsilon / 2)]}{\log 1/r}, 1 \}$ and $\epsilon' < \epsilon/2$. Then there is an algorithm that uses $\textsc{IndepOPT}(\lceil n/l \rceil, p, \epsilon')$ tests and finds the defective set with the error at most~$\epsilon  n$.
\end{theorem}
\begin{proof}
Consider the following algorithm:
\begin{enumerate}
    \item Let $l = \max \{\frac{\log [1/(1- \epsilon / 2)]}{\log 1/r}, 1 \}$. Partition the cycle into $\lceil n/l \rceil$ paths $P_1, P_2,\ldots, P_{\lceil n/l \rceil}$ of the same length $l$, except one path that may be shorter.
    \item For each path, choose one of its nodes at random and let the corresponding nodes be $v_{P_1}, v_{P_2}, \ldots v_{P_{\lceil n/l \rceil}}$.
    \item Use an $\textsc{IndepOPT}(\lceil n/l \rceil, p, \epsilon')$ algorithm (by Theorem~\ref{theorem:adaptive_indep} for adaptive or Theorem~\ref{theorem:non-adaptive_indep} for non-adaptive group testing) to find the defective items among $v_{P_1}, v_{P_2}, \ldots v_{P_{\lceil n/l \rceil}}$ where $\epsilon' < \frac{\epsilon}{2}$ and the probability of being defective equals  $p$.
    \item Assign  the state of  all the nodes in  $P_i$ as $v_i$ for all $i$.
\end{enumerate}

\noindent Note that for each $i$, the defectiveness probability of $v_i$ is $p$. The probability that $P_i$ is actually connected after a realization is $r^{l-1}$. So the probability that $P_i$ is not in the same state as $v_i$ is $1 - r^{l-1}$. Then assuming that we detect all $v_i$'s correctly, the error in $G$ is at most $\lceil n/l \rceil \cdot (1-r^{l-1}) \cdot l$. By replacing $l = \max \{\frac{\log [1/(1- \epsilon / 2)]}{\log 1/r}, 1 \}$, the error becomes less than $\epsilon n / 2$.
Moreover, we might also have $\epsilon'$ probability of error for the $v_i$'s (given the criteria set in \textsc{IndepOPT}), meaning that with probability $1- \epsilon'$, all the nodes are predicted correctly, and with probability $\epsilon'$ we have at least one mispredicted node, and at most $n$ mispredicted nodes. So the total error from this part is at most $\epsilon' n < \epsilon n /2$.
So the total error is at most $(\epsilon' + \epsilon /2) n < \epsilon n$ and we have the above theorem.
\end{proof}
\begin{corollary} \label{correl: cycle}
\textcolor{black}{Consider the case in which $p = c/n$ where $c$ is a constant. Let $\mathbf{X'}$ be the vector of candidate nodes. Then,  $H(\mathbf{X'}) = \frac{n}{l} H(p) =  O(\log n)/l$, where $H(.)$ is the binary entropy. Note that the average number of infected nodes in $\mathbf{X'}$ is $\mu = c/l$, hence by Theorem~\ref{theorem:adaptive_indep}, the number of tests is upper bounded by 
$$O(H(\mathbf{X'}) + \mu) = O(\frac{\log n + c}{l})= O(\frac{\log n \log 1/r}{\log [1/(1-\epsilon / 2)]})$$
$$\simeq O(\frac{\log n \log 1/r}{\epsilon}).$$
Note that when  correlations are strong, i.e., $r \geq 1- 1/ \log n $,  the algorithm does a constant number of tests, as expected. Note that using classic group testing without incorporating correlation we need $O(n H(c/n)) = O(\log n)$ tests. 
}
\end{corollary}

We now generalize the ideas to derive an upper bound when $G$ is a tree. We partition $G$ into $\lceil n/l \rceil$ subgraphs (which we sometimes refer to as groups) of $l$ nodes, find the probability of each subgraph being connected in a random realization, and then optimize it by choosing $l$.
{At a high level, we partition $G$ such that the nodes within a subgraph have small paths among each other. This is because  shorter paths remain connected in $G_r$ with higher probability, maximizing the probability of the nodes being in the same state. \textcolor{black}{Finding the probability of error is not straightforward here, because the subgraphs we form may not necessarily be connected in the original graph $G$ but could be connected through realization of other edges/nodes. }}

We first give a definition to formalize the number of nodes needed to make a subset of nodes connected.

\begin{definition}
Let $S \subseteq V$ be a subset of the nodes of graph $G$. The smallest connecting closure of $S$ is a subset $S' \subseteq V$ such that the induced graph over $S \cup S'$ is connected.
\end{definition}

\noindent For example, consider the graph $G_r$ in Figure~\ref{example:partition}. If $S = \{v_1, v_5 \}$, then the smallest connecting closure of $S$ is $\{ v_4 \}$, as by adding $v_4$ to $S$ we make S connected.

Note that if $G$ is a tree, then every connected subgraph of $G$ is also a tree. And the number of links in the induced graph over $S \cup S'$ is one less than the number of nodes in it. 

Now we provide a partition of nodes for trees such that for each subgraph, only a few additional nodes and thereby additional links will make them connected. Formally:

\begin{lemma} \label{lem:treepart}
Let $G$ be a tree. There is a partition of the graph into $\lceil n/l \rceil$ subgraphs each with $l$ nodes (one subgraph may have less than $l$ nodes), such that the number of nodes in the smallest connecting closure for each subgraph is less than or equal to $l$, for each $l \leq n$.
\end{lemma}

\begin{proof}

We prove the lemma by induction on the number of nodes of $G$. for $n = 1,2,3$, the statement is trivially true. Now suppose the lemma is true for any number of nodes less than $n$, we prove it for $n$.

We aim to find a set of nodes of size $l$ such that, first, by removing the set the graph remains connected, and second, the smallest connecting closure of the set has at most $l$ nodes.  Then by removing the aforementioned set and considering it as one of the subgraphs, we use induction hypothesis for the rest of the graph.

To do that, suppose the tree is hanged by an arbitrary node. Let $v$ be one of the deepest leafs, that is, any other node is at higher or equal level of $v$. Let $u$ be the first ancestor of $v$, such that the subtree rooted at $u$, including $u$, has $l$ or more nodes. If there is exactly $l$ nodes, then the subtree rooted at $u$ is the desired subgraph.

Now suppose that the subtree rooted at $u$ has more than $l$ nodes. Note that the number of nodes in the path from $v$ to $u$ is $l$ or fewer; otherwise $v$ would have had an ancestor lower than $u$ such that the subtree rooted at it would have at least $l$ nodes. Thus the distance from $v$ to $u$ is less than $l$. Now we form a subgraph $S$ with $l$ nodes and connecting closure of $l$ or fewer nodes. We progressively build $S$.   Starting with empty set $S$, we add the subtree of the child of $u$ that $v$ is a descendant of,  and the child itself; call this subtree $s_1$. Note that $s_1$ has $k < l$ nodes, otherwise $v$ would have had an ancestor lower than $u$ such that the subtree rooted at it would have at least $l$ nodes. Since $s_1$ is an entire  subtree rooted at a node,  $G$ remains connected even when $s_1$ is removed from it.

Consider $l_1 = l - k$. Recursively, do the same process for the other subtrees of $u$ with $l_1$ instead of $l.$ Note that $u$ has  subtrees other than $s_1$, as the subtree rooted at $u$ has more than $l$ nodes and $s_1$ has at most $l-1$ nodes. Then consider another subtree of $u$, called $s_2$. If $|s_2| \leq l_1$, we update $l_2 = l_1 - |s_1|$ and add $s_2$ to $S$ and continue with another subtree of $u$, which by the same argument exists. If $|s_2| > l_1$, then again we choose a deepest leaf of $s_2$ and proceed with the same process as before to find another group of nodes, i.e. we start with the deepest leaf and go up in the tree until it exceeds $l_1$, and repeat the procedure. Note that after moving to the subtrees of $u$, we disregard the rest of the graph, so $u$ is an ancestor of all the nodes we encounter next.

Again, for the next recursion, the subtree of the node that exceeds updated $l$ is an ancestor of the rest of the nodes. Let's call $u$ and other nodes that we make a recursion ``breaking point''. Then any pair of breaking points are ancestor and descendant, and all the nodes added to $S$ are subtrees of breaking points. So by connecting all the breaking points by a single path, which has length at most $l$, as the distance is less than or equal to $u$ to $v$, we connect $S$; so the smallest connecting closure of $S$ has $l$ or fewer nodes. More than that, we have only included some subtrees of $G$, so by removing $S$, $G$ remains connected and we can use the induction for the remaining tree.\\

To illustrate the algorithm, consider the tree at Figure~\ref{graph:tree} with $l = 5$. We start with $v$, which is a deepest leaf, move up and now the subtree is $\{7, v \}$ and we add node 7 to the subgraph as $l\geq2$. we move up again, and this time we can't add $u$ and its subtrees, as the size of the subtree rooted at $u$ would exceed $l$. So we update $l_1 = l - 2 = 3$ and proceeds with another subtree of $u$, let's say the right subtree containing node $8$. The size of the subtree at $8$ is one (only \{8 \}), so we can add it to the subgraph($l_1 \geq 1$) and update $l_2 = l_1 - 1 = 2$. Now we continue with the updated $l$ and the left subtree of $u$. Node 10 is a deepest leaf that we start with, move up to $u'$, but we can't add the subtree rooted at $u'$, as the size of the subtree rooted at $u'$ is bigger than $l_2$. So we add 10 to the set, update $l_3 = l_2-1 = 1$ and proceed with $u'$. Finally our updated $l$ is 1 and we only add 9 to the subgraph. So the final subgraph is $\{v, 7, 8, 9, 10  \}$, and we can connect all of them by adding $u$ and $u'$.

\begin{figure}
    \centering
    \input{Graph_tree}
    \caption{An example of the procedure in Lemma~\ref{lem:treepart}}
    \label{graph:tree}
\end{figure}
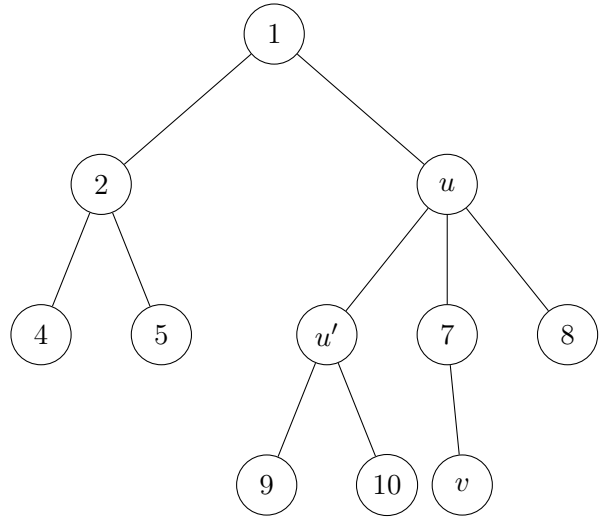

Now we've found $S$ that has the smallest connecting closure at most $l$, and includes only subtrees of $G$, so removing that does not disconnect the graph. Then we save $S$ as an aimed subgraph and use the induction hypothesis on the rest of the graph. Then other formed subgraphs have size $l$ (except one) and have small connecting closures. By adding $S$ to the subgraphs, we get the desired grouping of all nodes and the proof is complete. 
\end{proof}

Now we are ready to prove the upper bound for trees.
\begin{theorem} \label{theorem:treeup}
Consider a tree with $n$ nodes and let $l = \max \{\frac{\log [1/(1- \epsilon / 2)]}{2\log 1/r}, 1 \}$. Let $\epsilon' < \epsilon/2$. Then there is an algorithm that uses $\textsc{IndepOPT}(\lceil n/l \rceil, p, \epsilon')$ tests and finds the defective set with  at most $\epsilon n$ errors.~I.e., 
$$ \textsc{CrltOPT}(G, r, p, \epsilon) \leq \textsc{IndepOPT}(\lceil n/l \rceil, p, \epsilon').$$
\end{theorem}
\begin{proof}
Consider the following algorithm:
\begin{enumerate}
    \item By Lemma~\ref{lem:treepart}, partition the tree into $\lceil n/l \rceil$ subgraphs $g_1, g_2,\ldots, g_{\lceil n/l \rceil}$ of the same length $l$, one subgraph might be smaller than the other ones.
    \item For each group, choose one of its nodes at random and let them be $v_{g_1}, v_{g_2}, \ldots v_{g_{\lceil n/l \rceil}}$.
    \item Use an $\textsc{IndepOPT}(\lceil n/l \rceil, p, \epsilon')$ algorithm to find the defective set among $v_{g_1}, v_{g_2}, \ldots v_{g_{\lceil n/l \rceil}}$.
    \item Assign the  state of  all the nodes in $g_i$ as $v_i$, for all $i$.
\end{enumerate}
First, we calculate the probability that $g_i$ is connected. By Lemma~\ref{lem:treepart}, we know that each $g_i$ has the property that its   smallest connecting closure has 
$l$ or fewer nodes. Thus, together $g_i$ and its connected closure have at most $2l-1$ edges in $G$, and $g_i$ and its connected closure constitutes a connected subgraph in $G.$ Therefore, the probability of $g_i$ being connected in $G_r$ is at least the probability that the above edges are retained in $G_r$ which is at least $r^{2l-1} \geq r^{2l}.$ 
So the probability that $g_i$ is not in the same state as $v_i$ is at most $1 - r^{2l}$.
The rest of the proof revolves around proving that the total error is less than $\epsilon n$ as was done for cycle and this completes the proof. 
\end{proof}

\begin{corollary}
Corollary~\ref{correl: cycle} can be recovered for trees with an additional factor of 2.

\end{corollary}

\section{An Upper Bound for Graphs with More Edges: Grids and SBMs}
In this section, we focus on graphs that potentially have many edges. 
As the number of edges increases, the correlation between nodes increases even when $r$ is not large. As  mentioned earlier, we need to target those components that are more likely to appear in various realizations. 

We know that there is a threshold phenomenon in some edge-faulty graphs, meaning that when $r$ is below a threshold, there are many isolated nodes (and hence many independent tests are needed) and when $r$ is above that threshold, we have a giant component (and hence a single test suffices).  Most famously, this threshold is  $ \frac{\log n}{n}$ for Erd\H os-R\'enyi graphs.
For random $d$-regular graphs, also, \cite{goerdt1997giant} has shown that when a graph is drawn uniformly from the set of all $d$-regular graphs with $n$ nodes and then each edge is realized with probability $r$, $\frac{1}{d-1}$ is a threshold  almost surely. 


For the rest of this section, we first study a (deterministic) $4$-regular graph, known as the grid\footnote{There is a subtle difference worthwhile to mention here: The degree regularity does not hold on the boundaries of the grid.} and then provide near-optimal results for the stochastic block model.  When we consider (deterministic) $d$-regular graphs, we can't use the  results of \cite{goerdt1997giant} for random $d$-regular graphs because we can not be sure that the specific chosen graph is among the ``good'' graphs that constitute the almost sure result. So we need to develop new results on the number of connected components and the distribution on them for our purposes.

\subsection{The grid} \label{sec:grid}
We first formally define a grid. A grid with $n$ nodes and side length $\sqrt{n}$ is a graph where nodes are in the form of $(a,b): 1 \leq a,b \leq \sqrt{n}$. Node $(a,b)$ is connected to its four close neighbors (if exist), namely  $(a-1,b),(a+1,b), (a,b+1), (a,b-1)$. Border nodes (with $a \in \{1,\sqrt{n} \}$ or $b \in \{1, \sqrt{n} \}$) might have three or two neighbors. Note that the side length is defined by number of nodes on the side. 

\textcolor{black}{In order to derive a lower bound, we need to know the expected number of components in $G_r$ and equivalently the expected component size that nodes belong to \cite{alon2016probabilistic,goerdt1997giant}. To find the expected component size, we describe a random process that forms the components and analyze the expected stopping time. While this is well-known for ER graphs, to the best of our knowledge there are no bounds when the graph is not complete. We derive a lower bound on the expected number of components in a grid by analyzing the problem for 3-regular infinite trees. This approach can be generalized for any $d$-regular graph.}

Consider the following process. Pick a node $v \in V(G)$, mark it as processed, and let it be the root of a tree. For each $u \in V(G)$ that is not processed and is a neighbor of $v$, $uv$ is realized with probability $r$ and added as a child of $v$. The same process is repeated for each realized $u$ in a Breath First Search (BFS) order. When the process ends, there is a tree with root $v$, and the expected size of the tree is the expected size of the component that $v$ ends up in.
\begin{figure}
    \centering
    \input{grid}

    \caption{An example of the procedure described in Section~\ref{sec:grid}, starting with border node $v_{11}$.}
    \label{graph:grid}
\end{figure}
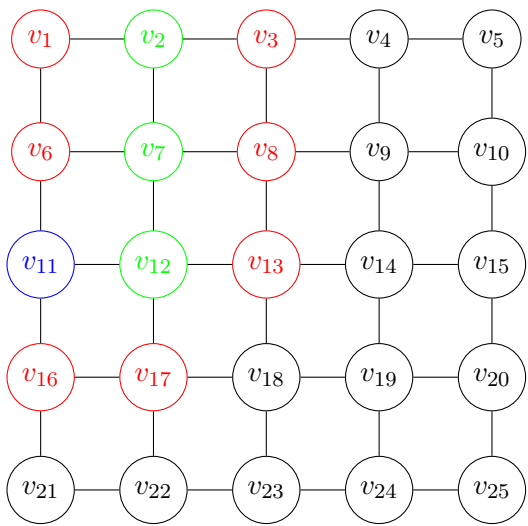

An example of this proecss is shown in Figure~\ref{graph:grid}. Node $v_{11}$ is the root (colored in blue),  the children that are realized are  marked in green, and the children that are not realized are marked in red. The component would be $\{v_{11}, v_{12}, v_7, v_2 \}$.

By repeating the process for each node that is not processed yet, we get a spanning forest. The expected number of components in the forest is the expected number of components in the original random graph. 

Here, the challenge is that we don't know the number of available (unprocessed) neighbors of a node. It highly depends on the previously chosen nodes, especially when $d$ is small, like in the grid. \textcolor{black}{Note that this is more complicated than the process for ER graphs, as the number of unprocessed neighbors in an ER graph is independent of the previously processed nodes (due to its homogeneous nature)}. We circumvent this issue by  analyzing an infinite regular tree process that effectively corresponds to a more connected graph and therefore leads to a lower bound on the expected number of connected components for the grid. 


\subsubsection{3-regular Trees}\label{sec:disttree}
Consider an infinite tree with root $v$ such that each node in the tree has three children. Consider the process where each edge is realized with probability $r$. Let $C(v)$ be the component that $v$ ends up in. The following lemmas approximates the distribution of $|C(v)|$.

\begin{lemma} \label{lem:distfinitietree}
Under the above process and for $t \in \mathbb{N}$, 
$$P(|C(v)| = t) = \frac{1}{2t+1} {3t \choose t} r^{t-1} (1-r)^{2t+1}$$
\end{lemma}
\begin{proof}
Let $T$ be an embedded tree in $G$ with $t$ nodes. In order to $T$ be realized in the process, all the edges in $T$ should be realized and the rest of edges that has a node in $T$ should not be realized. There are $t-1$ edges in $T$, and each node has three potential edges, so there are $2t+1$ edges that are not realized. So the probability that $T$ be realized is $r^{t-1}(1-r)^{2t+1}$.

Let $C_t$ be the number of trees with $t$ nodes and $v$ as the root. We have 
$$P(|C(v)| = t) = C_t \cdot r^{t-1} (1-r)^{2t+1}.$$

\noindent Note that $C_0 = C_1 = 1$. We find a closed form of $C_t$ by recursion. Node $v$ has three potential subtrees, where the sum of the size of the subtrees is $t-1$. We thus get the following recursion
$$C_t = \prod_{i+j+k = t-1, i,j,k\geq 0} C_iC_jC_k.$$

This recursion has the same initial points and the same recursion as second order Catalan numbers as follows:

\begin{lemma} \cite{aval2008multivariate} \label{theorem:catalan}
Order-$d$ Fuss-Catalan numbers that follows the recursion form 
$$C^d_n = \prod_{i_1+\cdots + i_d = n-1} C^d_{i_1} \cdots C^d_{i_d}$$ with $C^d_0 = C^d_1 = 1$, has the closed form $C^d_n = \frac{1}{(d-1)n + 1} \cdot {dn \choose n}$.
\end{lemma}

So the solution has the form $C_t = \frac{1}{2t+1} {3t \choose t}$ and the proof is completed.
\end{proof}

\begin{lemma}\label{lem:inftree}
Let $P_{\infty} = P(|C(v)| = \infty)$. Then,under the above process, 
\[   
P_{\infty} = 
     \begin{cases}
       0 &\quad {r\leq 1/3}\\
       \frac{3r - \sqrt{r(4-3r)}}{2r^2} &\quad\text{otherwise} 
     \end{cases}
\]
\end{lemma}
\begin{proof}
In order for $v$ to be in an infinite component, at least one of its children should be in an infinite component. There are three cases: (i) Either $v$ has one child, and this one is  infinite, which happens with probability $3r(1-r^2)P_{\infty}$; or 
(ii) $v$ has two children and at least one of them lies in an infinite component, which happens with probability $3r^2(1-r)(1 - (1-P_{\infty})^2)$; (iii) or $v$ has three children and at least one of them lies in an infinite component, which happens with probability $r^3(1 - (1 - P_{\infty})^3)$. So in total we have
$$P_{\infty} = 3r(1-r)^2P_{\infty} + 3r^2(1-r)(1 - (1-P_{\infty})^2) $$
$$+ r^3(1 - (1 - P_{\infty})^3).$$

{The solutions of this equation are $0$ and $P_{\infty} = \frac{3r \pm \sqrt{r(4-3r)}}{2r^2}$. Note that $P_{\infty} < 1$, so $\frac{3r + \sqrt{r(4-3r)}}{2r^2}$ is not valid. The relevant solution turns out to be dependent on whether $r > 1/3$ or $r \leq 1/3$. As a matter of fact,  when $r > 1/3$, the  probabilities of all finite components (as found in Lemma \ref{lem:distfinitietree}) do not add up to one. So for $r>1/3$, the correct solution is $\frac{3r - \sqrt{r(4-3r)}}{2r^2}$. When $r<1/3$, we have $\frac{3r - \sqrt{r(4-3r)}}{2r^2} < 0$, so zero is the correct solution and the lemma is proved.}
\end{proof}

\begin{theorem}\label{theorem:inftreecompexp}
When $r \leq 1/3$,  {the expected component size} is
\begin{align}
\mathbb{E}(|C(v)|) =\sum_{t=1}^{\infty} \frac{t}{2t+1} {3t \choose t} r^{t-1} (1-r)^{2t+1} \label{eq:theoremECN}\\
\simeq\frac{1-r}{r}  \sqrt{\frac{3}{4\pi}}\sum_{t=1}^{\infty} \frac{\sqrt t}{(2t+1) } (\frac{27}{4} r (1-r)^2)^t. \notag
\end{align}

\end{theorem}
\textcolor{black}{Note that when $r < 1/3$, then $\frac{27}{4} r (1-r)^2 < 1$ and hence the sum converges.}

\begin{proof}
The proof follows from Lemma~\ref{lem:distfinitietree} and Lemma~\ref{lem:inftree} and by using Stirling Approximation. 
\end{proof}

It is worthwhile to remark that the proof generalizes to general $d-$regular tree processes.
\begin{remark}
When the underlying graph is an infinite $d$-regular tree, under the defined process and for $t \in \mathbb{N}$, 
$$P(|C(v)| = t) = \frac{1}{(d-1)t+1} {dt \choose t} r^{t-1} (1-r)^{(d-1)t+1}$$
\end{remark}

\subsubsection{Lower Bound on the Expected Number of Components in a Grid} \label{subsec:lowgrid}

In the case of grid, we consider 3-regular trees, as after the root in grid, each child has at most 3 potential neighbors and if we choose a node in the border, the root also has at most 3 potential children, as illustrated in Figure~\ref{graph:grid}. So the 3-regular tree process that we analyzed corresponds to a more connected graph than the grid. Therefore its expected number of connected components that we found in  \eqref{eq:theoremECN} provides a lower bound on the expected number of connected components in the grid.

\textcolor{black}{Let $NC$ be the number of connected components. Note than $NC$ is a stopping time, as by knowing $C_1 + \cdots + C_{NC}$, where $C_i$ is the size of the $i$'th component, we can decide whether the process has finished or not. The random process in the $3-$regular tree is symmetric over all the nodes, so by $NC$ be stopping time, $\mathbb{E}[NC] = |V(G)|/\mathbb{E}[C(v)]$}. So by Theorem~\ref{theorem:inftreecompexp}, we immediately have the following result.

\begin{theorem} \label{theorem: lowerboundgrid}
For a grid with $n$ nodes and $r \leq 1/3$, the expected number of components is
$$\mathbb{E}(NC) =\frac{n} {\sum_{t=1}^{\infty} \frac{t}{2t+1} {3t \choose t} r^{t-1} (1-r)^{2t+1}}$$
$$\simeq \frac{n} {\frac{1-r}{r}  \sqrt{\frac{3}{4\pi}}\sum_{t=1}^{\infty} \frac{\sqrt t}{(2t+1) } (\frac{27}{4} r (1-r)^2)^t}.$$
\end{theorem}

\begin{corollary}
Using Lemma~\ref{lem:components} in conjunction with Theorem~\ref{theorem: lowerboundgrid}, any lower bound on the number of tests for $\frac{n} {\frac{1-r}{r}  \sqrt{\frac{3}{4\pi}}\sum_{t=1}^{\infty} \frac{\sqrt t}{(2t+1) } (\frac{27}{4} r (1-r)^2)^t} + O(1/n)$ independent nodes is also a lower bound on the number of tests on a grid under our model.

\end{corollary}

\subsubsection{An Upper Bound for the Grid}
{In this subsection, we provide an upper bound for the number of tests in a grid. At a high level idea, we partition the grid into subgrids and assume that each subgrid is connected, so we can consider one representative node per subgrid to test. In other words, the algorithmic idea is similar to the proof of Theorems~[\ref{theorem:cycle}, \ref{theorem:treeup}] where the graph $G$ was partitioned into subgraphs that were more likely to remain connected in $G_r$. But, here, we choose the subgraphs to be subgrids. In order to calculate the error, we need to compute the probability that a subgrid of length $k$ is connected, where $k$ is to be optimized later. We first estimate the probability that a subgrid becomes connected.}

\begin{lemma}\label{lem:subgridprob}
Let $P_k$ be the probability that a grid of length $k > 1$ becomes connected when each of its edge is realized with probability $r$. We have:
$$P_k \geq P_{k-1} r^{\Theta((1-r)k)}.$$
\end{lemma}
\begin{proof}
Consider the subgrid of length $k-1$ that contains the bottom-left corner node. Then the main grid consists of the subgrid and a path with $2k-1$ nodes,  where each node in the path has one edge to the subgrid. For example in Figure~\ref{graph:grid}, $k=5$, the subgrid is the grid with corner nodes $v_6, v_9, v_{24}, v_{21}$ and the path of with $9$ nodes is $v_1, v_2, v_3, v_4, v_5, v_{10}, v_{15}, v_{20}. v_{25}$.  With probability $P_{k-1}$ the subgrid is connected. Note that in expectation, $(2k-2)(1-r)$ edges in the path would be removed, and by Chernoff bound it is concentrated around its mean. Then, the path would be decomposed into $(1-r)2k \pm o(rk) = \Theta((1-r)k)$ subpaths with probability at least $1 - 1/k^{10}$. Each subpath has at least one edge to the subgrid, so each one is connected to the path with probability at least $r$. The probability that all of them connect to the subgrid then is at least $r^{(1-r)\Theta(k)}$ and the lemma is proved. 
\end{proof}
\begin{theorem}
Let $P_k$ be the probability that a grid of length $k > 1$ becomes connected when each of its edge is realized with probability $r$.We have:
$$P_k \geq r^{\Theta((1-r)k^2)} = e^{\Theta(\log(r) (1-r)k^2))}$$
\end{theorem}
\begin{proof}
The proof is done by replacing $P_{k-1}$ with $P_{k-2}$, and then $P_{k-2}$ with $P_{k-3}$ etc in Lemma~\ref{lem:subgridprob} and at last replacing $P_1 = 1$.
\end{proof}

We now partition the grid into subgrids of length $k$, $k$ will be set later, and consider a candidate node for each subgrid and do the independent group testing on candidate nodes. Now similar to Theorem~\ref{theorem:treeup}, by setting error probability of each subgraph small enough, that is  $1 - P_k < \epsilon / 2$, we get $k < \sqrt{\frac{\log(1 - \epsilon/1000)}{(1-r)\log r}}$. Then the error is less than $\epsilon n$ with at most $n / k^2$ independent node tests with error $\epsilon' < \epsilon /2$.

\subsection{An Optimal Algorithm for the Stochastic Block Model}

In this section, we study our model on SBM graphs. We apply the same techniques used in Erd\H{o}s-R\'enyi graphs to find the connectivity threshold to find the structure of the connected components.

A stochastic block model has $g$ clusters of size $k = n/g$, where any pair of nodes in the same cluster are connected with probability $q_1$, and any  pair of nodes in different clusters are connected with probability $q_2 < q_1$. After realizing each edge with probabilities $q_1$ and $q_2$, we have our graph $G$. Then based on our correlation  model, each edge remains in $G_r$ with probability $r$. So with probability $r_1 = rq_1$ an edge remains in the same cluster and with probability $r_2 = rq_2$ an edge remains between two different clusters. Here, we assume the size of the clusters are much bigger than $\log n$, i.e. $k \gg \log n$. We find the number of connected components based on $r_1$ and $r_2$ with high probability, for simplicity let's say 99\%. The probability can be improved with a slight change in the parameters.

\begin{theorem} \label{theorem:SBMstruc}
 \begin{itemize}

    \item If $r_1 \geq \frac{100\log n}{k}$ and $1 - (1-r_2)^{k^2} \geq \frac{100\log g}{g}$, then with high probability $G$ is connected. (first regime, one test needed)
    \item If $r_1 \geq \frac{100\log n}{k}$ and $1 - (1-r_2)^{k^2} \leq \frac{1}{100g}$, then with high probability each cluster is connected but most of the clusters are isolated. (second regime, $g$ independent tests needed)
    \item If $r_1 \leq \frac{1}{100k}$ and $r_2 \leq \frac{1}{100n}$, then with high probability $G_r$ has many isolated nodes. (third regime, $\Omega(n)$ independent tests needed)
    \item If $r_1 \leq \frac{1}{100k}$ and $r_2 \geq \frac{100\log n}{n}$, and $g>1$, then with high probability $G_r$ is connected. (fourth regime, one test needed)

\end{itemize}

\end{theorem}
\begin{proof}
    First, suppose $r_1 \geq \frac{100\log n}{k}$. A cut is a partition of nodes into two sets (parts), and its size is the number of nodes in the smaller set. We say a cut is disconnected if there is no edge between the sets. A cut of size $i \leq k/2$ in a single cluster has $i(k-i)$ potential edges between the parts. Then the probability that the specific cut is disconnected is 
    $$(1 - r_1)^{i(k-i)} \overset{(i)}{\leq} e^{-r_1 i(k-i)} \overset{(ii)}{\leq} e^{-100\log n i(1-i/k)} $$
    $$\overset{(iii)}{\leq} e^{-50 i \log gk} = (\frac{1}{gk})^{50i}.$$
    The first inequality, (i), is true because $1 - x \leq e^{-x}$ for $x\geq 0$, (ii) is true by $r_1 \geq \frac{100\log n}{k}$, and (iii) is true by $i \leq k/2$. Note that number of cuts of size $i$ is ${k \choose i}$, and by Union Bound, the probability that any cut of size $i$ becomes disconnected is at most $\sum_{i=1}^{k/2}(\frac{1}{gk})^{50i} {k \choose i}$. But ${k \choose i} \leq k^i$ by a simple counting argument, so the probability of a cut be disconnected is at most $\sum_{i=1}^{k/2}(\frac{1}{gk})^{50i} k^i < (\frac{1}{gk})^{48} = (1/n)^{48}$. So with probability $1 - (\frac{1}{n})^{48}$ a single cluster of size $k$ is connected.
    Again, by Union Bound, with probability at most $(\frac{1}{n})^{48} g < (\frac{1}{n})^{47}$ there is a disconnected cluster. So with probability $1 - (\frac{1}{n})^{47}$, all clusters are connected.
    
    Now if we assume all the clusters are connected, if there is an edge between two clusters, then those two clusters are connected. So if we consider a graph where the nodes represent the clusters and two nodes are connected if there is at least one edge between the corresponding clusters, then we need to understand the connectivity of the new graph. The probability that there is at least one edge between two clusters is $1 - (1-r_2)^{k^2}$, and again if this value is more than $\frac{100\log g}{g}$, then $G$ is connected with high probability. If $1 - (1-r_2)^{k^2} < \frac{1}{100g}$, then the probability that a cluster is isolated is more than $(1 - 1/(100g))^{g-1} \simeq e^{-1/100} \simeq 0.99$, so most of the clusters are isolated, which proves the first two parts of the theorem.
    
    If $r_1 \leq \frac{1}{100k}$, then with the same argument, with high probability most of the nodes in all clusters are isolated. If we also have $r_2 \leq \frac{1}{100n}$, then this means that most of the nodes don't have any neighbors outside of their cluster with high probability, so in total the graph has $\Omega(n)$ many isolated nodes, which proves the third part.
    
    Now suppose $r_2 \geq \frac{100\log n}{n}$, and we prove the last part of the theorem. We assume each cluster is empty, i.e. there is no edge in the cluster, and even when they're empty with $r_2 \geq \frac{100 \log n}{n}$, the graph is connected with high probability. Consider a cut in $G$ with $i \leq n/2$ nodes. Each node has $n-k$ potential neighbors in other clusters. So it has at least $n - k - i$ potential neighbors outside of its clusters and the chosen cut. Then, almost similar to the first part of the theorem, the probability that this cut is disconnected is at most 
    $$(1-r_2)^{i\cdot (n-k-i)} \leq e^{-r_2 \cdot i \cdot (n - k - i)}$$
    $$\leq e^{- 100 \log n \cdot i \cdot (\frac{n - k - i}{n})} \overset{(i)}{\leq} e^{- 100 \log n \cdot i \cdot (1/2 - 1/g)}$$
    $$= (\frac{1}{n})^{100i(1/2 - 1/g)}.$$
    Here, (i) is true by $i \leq n/2$ and $k/n = 1/g$. Again, there are ${n \choose i} \leq n^i$ cuts of size $i$. So the probability that any cut of size $i$ is disconnected is at most $n^i \cdot (\frac{1}{n})^{100i(1/2 - 1/g)} = (\frac{1}{n})^{100i(1/2 - 1/g - 0.01)}$. It is not hard to see that if $g > 2$, then $(\frac{1}{n})^{100i(1/2 - 1/g - 0.01)} = o(1/n^2)$. So the probability that any cut is disconnected is bounded by
    $$\sum_{i = 1}^{n/2} (\frac{1}{n})^{100i(1/2 - 1/g - 0.01)} \leq n \cdot o(1/n^2) = o(1/n).$$
    So in the case of $g > 2$, we've proved the last part of the theorem. If $g = 2$, for $i>2$, a cut of size $i$ has at most $i^2/4$ edges in the node set of size $i$, as the graph is bipartite and the number of edges in the set is maximized when $i/2$ nodes is chosen from each part of the graph. So the potential edges to the other side of the cut is at least $i(n-k) - i^2/4 = i(n-k- i/4) = i(\frac{n - i}{2}) \geq i \cdot 3n/8$, as $i \leq n/2$, and we can repeat the reasoning to prove that with high probability all cuts in this graph are connected. It is also easy to verify that when the cut is a single node or a pair of nodes, then the cut is disconnected with probability at most $o(1/n^4)$, and this completes the proof.
\end{proof}

Based on the previous theorem, we can now design a simple algorithm based on the parameters $r_1$ and $r_2$. In the first and the last regime, a single node is tested and the the result generalizes to all the nodes. In the second regime, we pick a candidate node from each cluster and perform independent group testing on them. The result of each candidate node generalizes for all the nodes in the correspondent cluster. Finally, in the third regime we perform independent group testing on the $n$ nodes of the graph.

\section{A Stronger Notion of Error}\label{sec:generalized}
So far, we have  focused on bounding the expected error of the algorithms, meaning the error is low (only) on average. But in many applications, we need to have a low error with high probability. As an example, let's say we need the error to be less than $e$. Under the weaker notion of error (bounding the average), we might have an error of $1.5 e$ half of the time, and for the other half have $.5 e$ error. So half of the time we don't satisfy the required error.  Similar to \cite{li2014group}, we introduce probabilistic error with relaxation on error-free prediction. Precisely, in \cite{li2014group} they wanted to find the defective set with probability $1- \delta$ such that all the nodes are correctly predicted, but we allow up to $\epsilon n$ mispredicted nodes.

We consider the following stronger notion of error: suppose that we want to have at most $\epsilon n$ mispredicted nodes with probability $1-\delta$, and for the $\delta$ other fraction we can have any number of mispredicted nodes. We refer to this notion of error as maximum error with parameters $\epsilon$ and $\delta$. This is a relaxation of the error compared to \cite{li2014group}, where $\epsilon = 0$, i.e. with probability $1-\delta$ we recover perfectly. Let $\textsc{CrltOPT}(G, r, p, \delta, \epsilon)$ be the expected number of tests in an optimal algorithm with maximum error with parameters $\epsilon$ and $\delta$. Parameters $r$ and $p$ are defined as in Section \ref{sec:model}.\\

In the following, we will provide lower bounds on the number of tests under the above notion of error. Then we will provide matching upper bounds for some of the graphs discussed so far. 

\subsection{Lower Bounds for the Stronger Error}

We first find a lower bound for $\textsc{CrltOPT}(G, r, p, \delta, \epsilon)$. Note that we  are no longer able to use  lower bounds of classic group testing directly, like Lemma~\ref{lem:components}, because an error in classical group testing (which affects the average error) might not be counted as an error in maximum error. In other words, in classical group testing, when an error happens with probability $\delta$, there is no guarantee on the number of mispredicted nodes, the error might be one or a constant or all the nodes, but only $\epsilon n$ mispredicted nodes are allowed in maximum error. So we need to find a lower bound {for the problem with maximum error definition and independent nodes} directly. We adapt the approach in \cite{li2014group} to derive  the following lower bound:

\begin{theorem}\label{Theorem: lowerboundnewerror}
 Let $\textsc{IndepOPT}(n,p,\delta,\epsilon)$ be the minimum number of tests for $n$ independent nodes under maximum error with parameters $(\delta, \epsilon)$. Then any Probabilistic Group Testing algorithm that,  with probability $1-\delta$, predicts all independent nodes but $\epsilon n$ of them correctly, needs $ n (1 - \delta) (H(p) - H(\epsilon))$ tests where $H$ is the binary entropy function. i.e.
$$\textsc{IndepOPT}(n,p,\delta,\epsilon) \geq n (1 - \delta) (H(p) - H(\epsilon)) - O(1).$$
\end{theorem}

\begin{proof}
The proof is a modified version of the proof of Theorem~\ref{theorem:lowerboundindep} in \cite{li2014group}. Let $\mathbf{X}$ be the vector of states of the nodes, $\mathbf{B}$ be the vector of the result of the group tests for a testing strategy of choice, and $\mathbf{Y}$ be the estimated states of the nodes. Then $\mathbf{X} \rightarrow \mathbf{B} \rightarrow \mathbf{Y}$ form a Markov chain. Thus, by data processing inequality, $I(\mathbf{X} ; \mathbf{B}) \geq I(\mathbf{X} ; \mathbf{Y})$. Also we have $I(\mathbf{X} ; \mathbf{B}) = H(\mathbf{B}) - H(\mathbf{B} \mid \mathbf{X}) \leq H(\mathbf{B}).$ Now, $H(\mathbf{B}) \leq 
\log _{2}|\mathbf{B}|$, where $|\mathbf{B}|$ indicates the number of possible values of the random vector $\mathbf{B}$. Since $\mathbf{B}$ represents the result vector, the number of possible values of this vector is \textcolor{black}{at most} $2^T$, where $T$ is the number of tests. Thus, $T {\color{black}{\geq}} \log _{2}|\mathbf{B}|$. Thus, combining the above inequalities, $T \geq I(\mathbf{X} ; \mathbf{Y}).$

Moreover,
$$
H(\mathbf{X})=H(\mathbf{X} \mid \mathbf{Y})+I(\mathbf{X} ; \mathbf{Y}).
$$ Thus, $T \geq H(\mathbf{X}) - H(\mathbf{X} \mid \mathbf{Y}).$ Since $X$ represents the state of $n$ independent nodes, each of which is defective with probability $p,$ $H(\mathbf{X}) = n H(p)$.
Thus,
\begin{equation}
T \geq n H(p) - H(\mathbf{X} \mid \mathbf{Y}). \label{eq:firstbound}
\end{equation}

We now obtain an upper bound for $H(\mathbf{X} \mid \mathbf{Y}).$
Define the error random variable $E$ such that
$$
E= \begin{cases}1, & \text { if } \|\mathbf{Y} - \mathbf{X}\|_0 > \epsilon n \\ 0, & \text { if } \|\mathbf{Y}-\mathbf{X}\|_0 \leq \epsilon n \end{cases}
$$
where $\|.\|_0$ is the number of non-zero elements in a vector.
We can bound the conditional entropy as follows
\begin{align}
H(\mathbf{X} \mid \mathbf{Y}) &=H(E, \mathbf{X} \mid \mathbf{Y}) \nonumber\\
&=H(E \mid \mathbf{Y})+\operatorname{Pr}[E=0] H(\mathbf{X} \mid \mathbf{Y}, E=0)\nonumber\\
&+\operatorname{Pr}[E=1] H(\mathbf{X} \mid \mathbf{Y}, E=1) \nonumber\\
& \leq 1+(1-\delta) H(\mathbf{X} \mid \mathbf{Y}, E=0) + \delta H(\mathbf{X})\nonumber\\
& \leq 1+(1-\delta) H(\mathbf{X} \mid \mathbf{Y}, E=0) +  n \delta H(p).
\label{eq: bigbound}
\end{align}
We now upper bound $H(\mathbf{X} \mid \mathbf{Y}, E=0)$:

\begin{align}
&H(\mathbf{X} \mid \mathbf{Y}, E=0) =\nonumber\\
& \sum_i \operatorname{Pr}[Y = y_i| E = 0] H(\mathbf{X}\mid Y = y_i, E = 0] \nonumber\\
&\leq \sum_i \operatorname{Pr}[Y = y_i]  \log c'{n \choose \epsilon n}\nonumber\\
& = \log {n \choose \epsilon n} + c' \simeq n H(\epsilon). \label{eq: toreplace1}
\end{align}
To obtain the inequality, we note that given  that $E = 0$, \textcolor{black}{there are at most $\epsilon n$ bits in which $\mathbf{X}$ differs from $\mathbf{Y}$. Thus, given a value of $\mathbf{Y}$ and that $E = 0$, there are at most $\sum_{i =0}^{\epsilon n} {n \choose i} \leq c' {n \choose \epsilon n}$ values of $\mathbf{X}$, where $c'$ is a constant.} The result follows by recalling that the entropy for any random variable with $r$ values is \textcolor{black}{at most} $\log r$. The Theorem follows by putting together (\ref{eq:firstbound}), (\ref{eq: bigbound}) and (\ref{eq: toreplace1}). 

\end{proof}
Analogous to Lemma~\ref{lem:components}, we immediately get the following lemma:
\begin{lemma}\label{lemma: newlowbound}
Let $C(G_r)$ be the number of connected components in $G_r$. Then, under  a maximum error target, we have
\begin{align}
C(G) (1 - \delta) (H(p) - H(\epsilon))  \leq \textsc{CrltOPT}(G, r, p, \delta, \epsilon). \label{eq: old}
\end{align}
\end{lemma}

\noindent \textcolor{black}{Using Lemma~\ref{lemma: newlowbound} along with our concentration results on the number of connected components for cycles, trees and grids, we find lower bounds on the number of tests for our maximum error criteria.
Recall that the lower bound for the average error $\epsilon$ has the form of $C(G) (1 - \epsilon n) H(p)$ for a graph $G$, while under the stronger notion of error we have $C(G) (1-\delta) (H(p) - H(\epsilon))$. For the regime where $\epsilon = c/n, c<1,$ the lower bound for average error and the stronger error simplify to $C(G)(1-c)H(p)$ and $C(G)(1-\delta) (H(p) - \frac{\log n}{n}) \simeq C(G)(1-\delta)H(p)$, respectively, and when $\delta < c$ we get an improvement.}  \\

\textcolor{black}{As discussed, there is  a gap between our lower and upper bounds. We next show that the lower bound can be improved by capturing the underlying topology of the graph more heavily.}

\subsection{An Improved Lower Bound for the Star Graph}

\textcolor{black}{Theorem~\ref{Theorem: lowerboundnewerror} does not depend on the underlying graph $G$. However, if we take $G$ into account, we can infer information about the state of the nodes, at least for some graphs. Here, we give an improved lower bound for  star graphs, where there is a node with degree $n-1$ and all the other nodes are leaves. 
}
\begin{theorem}\label{Theorem: starlowerbound}
When the underlying graph $G$ is a star, we have
\begin{align}
&n(1-\delta) \left( H(r) + (1-r) H(p) - H(\epsilon) - 1 + p(1-p) (1-r) 
H(r') + o(1)\right)- O(1) \nonumber\\
&\leq \textsc{CrltOPT}(G, r, p, \delta, \epsilon)\label{eq:new}
\end{align}

where $r' = \frac{r}{r+(1-r)(p^2+ (1-p)^2)}$. 
\end{theorem}

\begin{proof}
Let $\mathbf{X}$, $\mathbf{B}$ and $\mathbf{Y}$ be defined same as in the proof of Theorem~\ref{Theorem: lowerboundnewerror}. Let $\mathbf{G}$ be a random binary vector where the $i$'th coordinate is 1 iff the $i$'th edge of $G$ is realized (with probability $r$). Then $(\textbf{X},\textbf{G}) \rightarrow \textbf{B} \rightarrow \textbf{Y}$ forms a Markov chain. We are interested in $I(\X,\G ; \Y)$ because

$$I(\X, \G ; \Y) \leq I(\X, \G ; \B) \leq H(\B) \leq \log |\B| = T$$
where $T$ is the number of test. Write 
\begin{align}
I(\X, \G ; \Y) = H(\X, \G) - H(\X, \G | \Y) \label{eq:1}.
\end{align}

We now bound $H(\X, \G | \Y)$. Let $E = 1$ if $\|X-Y\|_0 > \epsilon n$ and $E = 0$ otherwise, same as Theorem~\ref{Theorem: lowerboundnewerror}. Then
\begin{align}
 H(\X, \G | \Y) &= H(\X,\G, E | \Y) \nonumber\\
& = H(E | \Y) + \delta H(\X, \G | \Y, E = 1) + (1-\delta)H(\X, \G | \Y , E = 0). \label{eq:lasteq}
\end{align}
By writing $H(E | \Y)\leq 1$ and $H(\X, \G | \Y, E = 1) \leq H(\X, \G)$ and replacing Eq~\eqref{eq:lasteq} in  Eq~(\ref{eq:1}), we get
\begin{align} \label{eq:2}
& I(\X, \G ; \Y) \geq 
 (1-\delta) (H(\X, \G) - H(\X,\G | \Y, E = 0))-1.
\end{align}
As $G$ is a tree and every component is independently infected with probability $p$, we can write $$H(\X,\G) = H(\G) + H(\X | \G) \simeq n H(r) + (1-r)n H(p) = n\left(H(r) + (1-r)H(p)\right).$$
Now we bound $H(\X,\G | \Y, E = 0) = H(\X | \Y , E= 0) + H(\G|\X)$. For the first term, in Theorem~\ref{Theorem: lowerboundnewerror} we found 
\begin{align}
H(\X | \Y , E = 0) \lesssim n H(\epsilon). \label{eq:Heps}    
\end{align}

\noindent To bound the second term, note that the underlying graph $G$ is a star and with probability $1 - 1/n^{2}$, the number of nodes in a different state than the center is $2np(1-p)(1-r) \pm o(n)$, so the contribution outside of this range to $H(\mathbf{G}|\mathbf{X})$ is only $o(1/n)$. The edges connected to such nodes can not be realized, so $n(1 - 2p(1-p)(1-r))$ edges are still uncertain, and knowing that the two end-points are in the same state, each of such edges are realized with probability $r' = \frac{r}{r+(1-r)(p^2+ (1-p)^2)}$. So
\begin{align}
H(\G | \X) \leq n\left(1 - 2p(1-p)(1-r)\right) H(r'). \label{eq:HGX}
\end{align}

Again, by replacing all in Eq~(\ref{eq:2}), we get  a lower bound on the number of tests:
\begin{align}
T&{\geq} I(\X, \G ; \Y) \nonumber\\
& \stackrel{a}{\geq} (1-\delta) (H(\X, \G) - H(\X,\G | \Y, E = 0)) \nonumber\\
& = (1-\delta) (H(\G) + H(\X | \G) - H(\X | \Y , E= 0) - H(\G|\X)) -1\nonumber\\
&\stackrel{b}{\geq}(1-\delta) (nH(r) + n(1-r)H(p) - nH(\epsilon)   - H(\G|\X) + o(n H(r))) \label{eq:beforestar} -1\nonumber\\
& \stackrel{c}{\geq} (1-\delta)(n (H(r) + (1-r)H(p)) - (nH(\epsilon) \nonumber +n(1 - p(1-p)(1-r)) H(r'))+ o(n)) -1\nonumber\\
&= n(1-\delta) (H(r) + (1-r) H(p) - H(\epsilon) - 1 \nonumber + p(1-p) (1-r) H(r') + o(1)) -O(1)
\end{align}

Where $a$ is by Eq~\ref{eq:2}, $b$ is by Eq~\ref{eq:Heps} and $c$ is by Eq~\ref{eq:HGX}.

\end{proof}
\begin{figure}
    \centering
    \includegraphics[scale = .3]{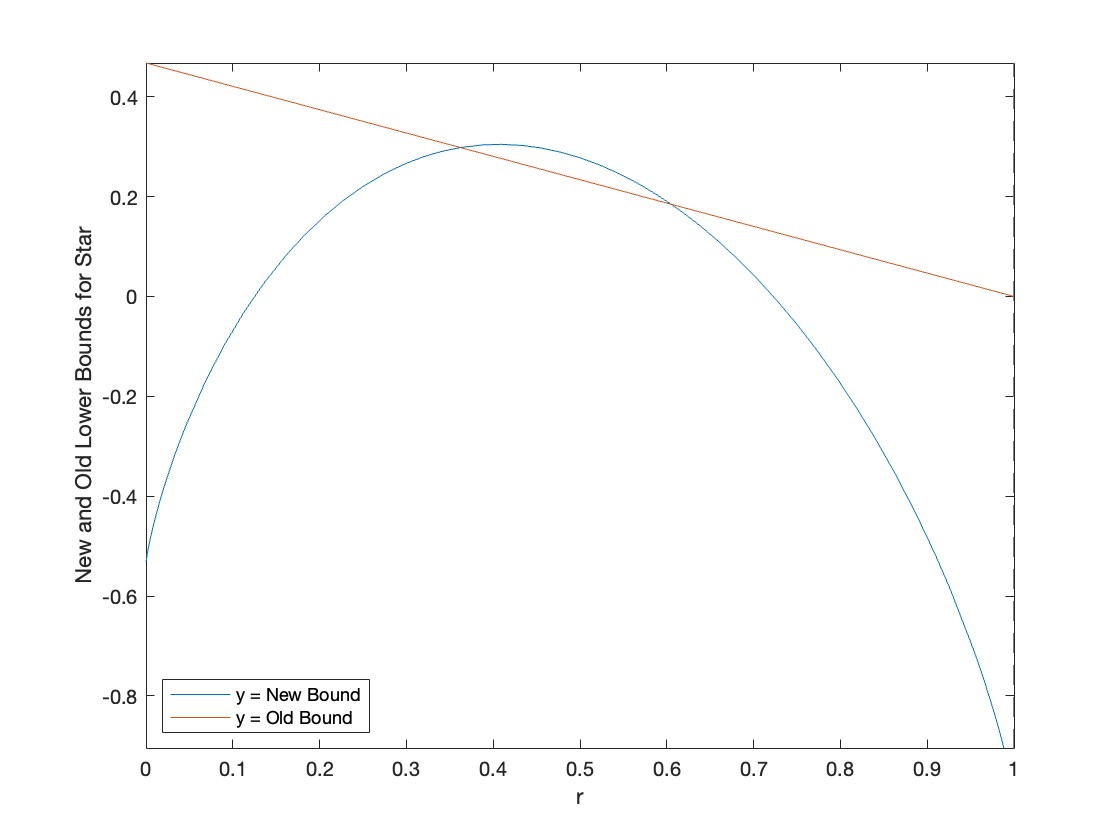}
    \caption{A comparison of the new lower bound and the old upper bound for $p = 0.1$ and $\epsilon = 0.0001$}
    \label{fig:starlow}
\end{figure}

\textcolor{black}{We compare the lower bound in \eqref{eq:new}(Star-specific bound) with the lower bound in \eqref{eq: old} (generic bound) in Figure~\ref{fig:starlow}. One sees  that the lower bound in \eqref{eq:new} is strictly larger than \eqref{eq: old} for a range of $r$ around $r=\frac{1}{2}$. This range corresponds to cases where uncertainty about the edge set of $G_r$ is high, and therefore knowledge about the existence of (or lack thereof) an edge is very informative as captured in our way of upper bounding $H(\mathbf{G}|\mathbf{X})$.
Even though the improvement offered by the lower bound in \eqref{eq:new} is relatively small, it suggests that generic lower bounds are likely to be loose for our correlated group testing problem and the structure of the underlying correlation graph $G$ need to be considered in order to obtain tighter lower bounds.}


\subsection{Upper Bounds}

We now find  upper bounds similar to those that we provided in Sections~\ref{sec:cycleandtree} and \ref{sec:grid}. Under an average error target, we partitioned the graph and computed the incurred average error. Under maximum error targets, however, we don't know what fraction of the realizations with average error $\epsilon n$ actually has less than $\epsilon n$ mispredicted nodes, so we can't use those results and need to prove concentrations around the average.
For cycles and grids, the subgraphs we designed were independent of each other, in the sense that the connectivity of a subgraph would not change the probability of connectivity of the other subgraphs. Hence by using Hoeffding's bound we prove the following theorem for a maximum error target.

\begin{theorem}\label{theorem:cyclenew}
Consider a cycle of length $n$ and let $l = \max \{\frac{\log [1/(1- \epsilon / 2)]}{\log 1/r}, 1 \}$, $\delta > 2e^{-\Theta(\frac{\epsilon^2 n \log(1/r)}{\log( \frac{1}{1 - \epsilon/2})})}$, and $\epsilon' < \delta/2$. There is an algorithm that uses $\textsc{IndepOPT}(\lceil n/l \rceil, p, \epsilon')$ tests and finds the defective set with maximum error with  parameters $\epsilon$ and $\delta$, i.e.
$$\textsc{CrltOPT}(Cycle, r, p, \delta, \epsilon) \leq \textsc{IndepOPT}(\lceil n/l \rceil, p, \epsilon')$$

\end{theorem}
\begin{proof}
We use the same algorithm and the same subgraphs as in Theorem~\ref{theorem:cycle}. Recall that the probability of one subgraph not being connected is $r^{l-1}$ and the average error is $\epsilon n /2$.
    The number of mispredicted nodes in each subgraph is in $[0,l]$ and they are independent of each other, meaning the connectivity of a group does not change the connectivity of another group, so by Hoeffding's bound, with probability at least $1-\exp[-\Theta(\frac{\epsilon^2 n \log(1/r)}{\log( \frac{1}{1 - \epsilon/2})})] > 1 - \delta/2$, the error is at most $\epsilon n$. Also with probability more than $1-\epsilon' > 1-\delta/2$ all candidates are predicted correctly, so with probability more than $1-\delta$ the classic group tests detect all the defective nodes with no error, and assuming this, the error on subgraphs is less than $\epsilon n$ and we're done.
\end{proof}
The same reasoning works for {subgrids of a grid and leads to the same upper bound with the parameters $\delta > 2e^{-\Theta(\frac{\epsilon^2 n \log(1/r)}{\log(\frac{1}{1 - \epsilon/2})})}$, and $\epsilon' < \delta/2$.}  But for trees, we can't use Hoeffding's bound because the connectivity of a subgraph can affect the others, for instance the absence of an edge might make several groups disconnected, hence the groups are not independent anymore. This dependency violates the conditions  we need for applying Hoeffding's inequality. In order to fix the issue, we use the node exposure martingales process with a proper graph function definition to prove the desired concentration.

Before providing the new theorem for trees, we need the following lemma to use the concentration lemma for node exposure martingale defined in Section~\ref{sec:lowerbound}.
\begin{lemma} \label{lemma:order}
Let $g_1, g_2, \ldots, g_{\lceil n/l \rceil}$ be the subgraphs formed in Lemma~\ref{lem:treepart}. Let $f(H)$ be the number of connected subgraphs $g_i$ in  graph $H$. Then there is an order of node exposure such that at each step, the value of $f$ does not change by more than one.
\end{lemma}
\begin{proof}
Let $g_1, g_2, \ldots, g_{\lceil n/l \rceil}$ be the subgraphs formed in Lemma~\ref{lem:treepart} in this order. Consider the following order of node exposure: we first expose all the nodes in the last subgraph, $g_{{\lceil n/l \rceil}}$, in some order. Then expose the nodes in the subgraph before that, $g_{(\lceil n/l \rceil -1)}$, and so on until the nodes in the last subgraph, $g_1$, are exposed. 

Consider a node $v$ in subgraph $g_i$ when it is exposed. Note that no subgraph $g_j, j > i$ can become connected after exposing $v$, as by construction of subgraphs, all the nodes of connecting closure of $g_j$ lies in $g_k, k>j$. Also, no subgraph $g_j, j < i$ can become connected, because neither of $g_j$'s nodes are exposed yet. So $f(H)$ can potentially only change by one, and for the last node of $g_i$ to make $g_i$ connected, and the lemma is proved.
\end{proof}
The above lemma allows us to use Azuma's inequality for groups made for trees as described in the following theorem.

\begin{theorem}\label{theorem:treenew}
Consider a tree with $n$ nodes and  let $l = \max \{\frac{\log [1/(1- \epsilon / 2)]}{2\log 1/r}, 1 \}$, $\delta > 2e^{-\Theta(\frac{\epsilon^2 n \log(1/r)}{\log(\frac{1}{1 - \epsilon/2})})}$, and $\epsilon' < \delta/2$. Then there is an algorithm that uses $\textsc{IndepOPT}(\lceil n/l \rceil, p, \epsilon')$ tests and finds the defective set with maximum error with  parameters $\epsilon$ and $\delta$, i.e.

$$\textsc{CrltOPT}(Tree, r, p, \delta, \epsilon) \leq \textsc{IndepOPT}(\lceil n/l \rceil, p, \epsilon')$$

\end{theorem}

\begin{proof}
We use the algorithm and the same subgraphs introduced in Theorem~\ref{theorem:cycle}. Recall from the proof of Theorem \ref{theorem:treeup} \ that the probability of one group not being connected is $r^{2l}$ and the average error is $\epsilon n /2$. Now we can't use the Hoeffding's bound to show concentration around the average. Instead, we use a node exposure martingale to prove the concentration. Note that Theorem~\ref{theorem:alonconct} for node exposure martingale only works when the graph function is node Lipschitz, meaning when $H_1$ and $H_2$ only differ in one node, $|f(H_1) - f(H_2)| \leq 1$. But if we can find an order of nodes $v_1,\ldots, v_n$ exposed such the graph $H_i$ on first $i$ nodes satisfies $|f(H_{i+1}) - f(H_i)| \leq 1 $, then we can still use  Azuma's inequality. Let $f(H)$ be the number of connected subgraphs $g_i$in $H$. Then by Lemma~\ref{lemma:order} there is an order such that $|f(H_{i+1}) - f({H_i})| \leq 1$, and we can use the concentration theorem (Theorem~\ref{theorem:alonconct}) for the number of connected groups in the random graph $G_r$. We now have the same error concentration as in Theorem~\ref{theorem:cyclenew} for error (equivalent to the Hoeffding's bound), hence we can repeat the argument to complete the proof.
\end{proof}

\section{Conclusion}
In this paper, we consider group testing strategies for identifying defective items when the defects of different nodes are correlated. The correlation is modeled through an underlying graph in which the degree of correlation between the defects in the items depends on the distance between the corresponding nodes in the graph. We relate the problem of design of testing strategies in presence of such correlation to that when the defects are independent. We subsequently obtain testing strategies in terms of those already known for independent defects for a large class of underlying graphs, namely trees, cycles, grids, and stochastic block models. This provides an upper bound for the number of tests needed to ensure the desired error bounds. We also obtain fundamental limits ie lower bounds  on the minimum number of tests required to ensure the same error bounds using bounds already known when the defects are independent. The bounds are obtained through a novel combination of edge exposure Martingale theory and graph partition techniques. 

We now describe some directions for future research stemming from the above results. There is a gap between the upper and lower bounds, which may be because there is scope of improvement in the lower bound, possibly utilizing the specific structures of the underlying correlation graphs. Another important area for future work is the development of testing strategies for general graphs. Towards that end one may envision the partition of the graph into structures such as trees, cycles, grids, stochastic block models etc for which we have identified in this paper testing strategies with guarantees on error rates and the number of tests required. Furthermore, designing group testing strategies when the defects dynamically evolve over time over graphs in question remains open. This problem has started receiving attention  \cite{srinivasavaradhan2022dynamic, arasli2021group}.

\newpage

\bibliographystyle{alpha}
\bibliography{ref}

\end{document}

%% file: headers.tex
\usepackage{amsmath,amsfonts,amsthm,amssymb}
\usepackage{color}
\usepackage[colorlinks=true]{hyperref}
\usepackage[procnumbered,ruled,vlined,linesnumbered]{algorithm2e}
\usepackage{amsfonts}

\newtheorem{theorem}{Theorem}[section]
\newtheorem{corollary}[theorem]{Corollary}
\newtheorem{lemma}[theorem]{Lemma}

\newtheorem{definition}[theorem]{Definition}
\newtheorem{remark}[theorem]{Remark}

\newenvironment{fminipage}%
  {\begin{Sbox}\begin{minipage}}%
  {\end{minipage}\end{Sbox}\fbox{\TheSbox}}

\DontPrintSemicolon
\SetKw{KwAnd}{and}
\SetProcNameSty{textsc}
\SetFuncSty{textsc}
\SetKwInOut{Input}{Input}
\SetKwInOut{Output}{Output}

%% file: graph_partition1.tex
\begin{subfigure}{.5\textwidth}
\centering
\begin{tikzpicture}[node distance={15mm}, main/.style = {draw, circle}]
\node[main] (1) {$v_1$};
\node[main] (2) [below right of=1] {$v_2$}; 
\node[main] (3) [above right of=1] {$v_3$};
\node[main] (4) [above left of=1] {$v_4$}; 
\node[main] (5) [below left of=1] {$v_5$};
\draw (4) -- (3);
\draw (4) -- (1);
\draw (4) -- (5);
\draw (5) -- (1);
\draw (5) -- (2);
\draw (3) -- (1);
\draw (3) -- (2);
\draw (2) -- (1);

\end{tikzpicture}
\caption{Graph $G$ before realization.} \label{fig:P1}

\end{subfigure}%

%% file: graph_partition2.tex
\begin{subfigure}{.5\textwidth}
\centering
\begin{tikzpicture}[node distance={15mm}, main/.style = {draw, circle}]
\node[main] (1) {$v_1$};
\node[main] (2) [below right of=1] {$v_2$}; 
\node[main] (3) [above right of=1] {$v_3$};
\node[main] (4) [above left of=1] {$v_4$}; 
\node[main] (5) [below left of=1] {$v_5$};
\draw (4) -- (1);
\draw (4) -- (5);
\draw (3) -- (2);

\end{tikzpicture}
\caption{Graph $G$ after a realization with 3 edges.} \label{fig:P2}

\end{subfigure}

%% file: star.tex
\begin{tikzpicture}[node distance={10mm}, main/.style = {draw, circle}]
\node[main, fill = pink] (10) {$C$};

\foreach \a in {1,2,...,9}{
\ifthenelse{\a < 5}{\def\mycol{pink}}{\def\mycol{white}}
\draw  (\a*360/9: 2cm) node[main, fill = \mycol] (\a) {\a};
}
\draw (10) -- (1);
\draw (10) -- (2);
\draw (10) -- (3);
\draw (10) -- (4);
\draw (10) -- (5);
\draw (10) -- (6);
\draw (10) -- (7);
\draw (10) -- (8);
\draw (10) -- (9);

\end{tikzpicture}

%% file: Graph_tree.tex
\begin{tikzpicture}[node distance={50mm}, main/.style = {draw, circle, minimum size =.8cm}, scale = .2]

\node [main] {1} [sibling distance = 8cm, level distance = 10cm]
    child {node [main, xshift = -1.5cm] {2}
    child {node [main] {4}}
    child {node [main] {5}}
    } 
    child {node [main, xshift = 1.5cm] {$u$}
    child {node [main] {$u'$}
    child {node [main] {9}}
    child {node [main] {10}}
    }
    child {node [main] {7}
    child {node [main, xshift = .2cm] {$v$}}
    }
    child {node [main] {8}}
    };

\end{tikzpicture}

%% file: grid.tex
\begin{subfigure}{.5\textwidth}
\centering
\begin{tikzpicture}[node distance={15mm}, main/.style = {draw, circle}]

\node[main, red] (1) {$v_1$};
\node[main, green] (2) [right of=1] {$v_2$}; 
\node[main, red] (3) [right of=2] {$v_3$}; 
\node[main] (4) [right of=3] {$v_4$}; 
\node[main] (5) [right of=4] {$v_5$}; 

\node[main, red] (6)[below of=1] {$v_6$};
\node[main, green] (7) [right of=6] {$v_7$}; 
\node[main, red] (8) [right of=7] {$v_8$}; 
\node[main] (9) [right of=8] {$v_9$}; 
\node[main] (10) [right of=9] {$v_{10}$}; 

\node[main, blue] (11)[below of=6] {$v_{11}$};
\node[main, green] (12) [right of=11] {$v_{12}$}; 
\node[main, red] (13) [right of=12] {$v_{13}$}; 
\node[main] (14) [right of=13] {$v_{14}$}; 
\node[main] (15) [right of=14] {$v_{15}$}; 

\node[main, red] (16)[below of=11] {$v_{16}$};
\node[main, red] (17) [right of=16] {$v_{17}$}; 
\node[main] (18) [right of=17] {$v_{18}$}; 
\node[main] (19) [right of=18] {$v_{19}$}; 
\node[main] (20) [right of=19] {$v_{20}$}; 

\node[main] (21)[below of=16] {$v_{21}$};
\node[main] (22) [right of=21] {$v_{22}$}; 
\node[main] (23) [right of=22] {$v_{23}$}; 
\node[main] (24) [right of=23] {$v_{24}$}; 
\node[main] (25) [right of=24] {$v_{25}$}; 

\draw (1) -- (2);
\draw (2) -- (3);
\draw (3) -- (4);
\draw (4) -- (5);

\draw (6) -- (7);
\draw (7) -- (8);
\draw (8) -- (9);
\draw (9) -- (10);

\draw (11) -- (12);
\draw (12) -- (13);
\draw (13) -- (14);
\draw (14) -- (15);

\draw (16) -- (17);
\draw (17) -- (18);
\draw (18) -- (19);
\draw (19) -- (20);

\draw (21) -- (22);
\draw (22) -- (23);
\draw (23) -- (24);
\draw (24) -- (25);

\draw (1) -- (6);
\draw (6) -- (11);
\draw (11) -- (16);
\draw (16) -- (21);

\draw (2) -- (7);
\draw (7) -- (12);
\draw (12) -- (17);
\draw (17) -- (22);

\draw (3) -- (8);
\draw (8) -- (13);
\draw (13) -- (18);
\draw (18) -- (23);

\draw (4) -- (9);
\draw (9) -- (14);
\draw (14) -- (19);
\draw (19) -- (24);

\draw (5) -- (10);
\draw (10) -- (15);
\draw (15) -- (20);
\draw (20) -- (25);

\end{tikzpicture}

\end{subfigure}